\documentclass[12pt,onecolumn]{IEEEtran}
\IEEEoverridecommandlockouts          % This command is only
\usepackage{balance}
\usepackage{graphics}
\usepackage{graphicx}
\usepackage{amssymb}
\usepackage{verbatim}
\usepackage{amsxtra}
\usepackage{epsfig}
\usepackage{subfigure}
\usepackage{amsmath}
\usepackage{latexsym}
\usepackage{ifthen}
\usepackage{psfrag}
\usepackage{color}
\usepackage{rgsEnvironments}
\usepackage{rgsMacros}
\usepackage{ulem}

\usepackage{pstool}

\renewcommand{\Re}{\mathbb{R}}

\newcommand{\KK}{{\cal K}}
\renewcommand{\K}{{\cal M}}
\renewcommand{\P}{{\cal P}}

\renewcommand{\defset}[2]{\bigbrace{#1\ : \ #2 }}
\newcommand{\Cp}{\reals^{n_p}}
\renewcommand{\C}{\Phi}

\begin{document}

\title{\bf
Robust Supervisory Control for Uniting Two Output-Feedback Hybrid Controllers with Different Objectives}
\author{Ricardo G. Sanfelice\thanks{Department of Aerospace and Mechanical Engineering, University of Arizona
1130 N. Mountain Ave, AZ 85721, USA (Tel: 520-626-0676; e-mail: sricardo@u.arizona.edu). Research partially supported by the National Science Foundation under CAREER Grant no. ECS-1150306 and by the Air Force Office of Scientific Research under Grant no. FA9550-12-1-0366.} 
and 
Christophe Prieur\thanks{Department of Automatic Control, Gipsa-lab, 961 rue de la Houille Blanche, BP 46,
38402 Grenoble Cedex, France (e-mail: christophe.prieur@gipsa-lab.grenoble-inp.fr). Research partially supported by HYCON2 Network of Excellence ``Highly-Complex and Networked Control Systems,'' grant agreement 257462.}}

\maketitle

\begin{abstract}
The problem of robustly, asymptotically stabilizing a point (or a set) with
two output-feedback hybrid controllers is considered.
These control laws may have different objectives, e.g., 
the closed-loop systems resulting with each controller
may have different attractors.
We provide a control algorithm that combines the
two hybrid controllers to accomplish the stabilization task.
The algorithm consists
of a hybrid supervisor that, based on the values of plant's outputs and 
(norm)
state
estimates, selects the hybrid controller that should be applied to the plant.
The accomplishment of the stabilization task relies on an output-to-state
stability property induced by the controllers, which
enables the construction of an estimator for the norm of the plant's state.
The algorithm
is motivated by and applied to robust, semi-global stabilization problems
uniting
two
controllers.
\end{abstract}

\section{Introduction}
\label{sec:intro}

\subsection*{Background and Motivation}
Many control applications cannot be solved by means of a single
state-feedback controller. As a consequence, control algorithms
combining more than one controller have been thoroughly investigated
in the literature. Particular attention has been given to the problem of uniting local and
global controllers, in which two control laws are used: one that is
supposed to work only locally, perhaps guaranteeing good
performance, and another that is capable of steering the system
trajectories to a neighborhood of the operating point, where the
local control law works. Different strategies are possible to tackle this problem. In \cite{PanEzalKrenerKokotovic_TAC_01},
this problem is solved by 
patching
 together a local optimal controller and a global controller designed using backstepping. In \cite{MorinMurrayPraly:NOLCOS:98}, a static time-invariant controller was designed 
 by smoothly blending
 global and local controllers.  In \cite{AndrieuPrieur10}, 
two control-Lyapunov functions are combined
to design a global stabilizer for a class of nonlinear systems.

The use of discrete dynamics may be necessary when piecing together 
local and global controllers (e.g., see the example in \cite{Prieur01}, 
where 
local and global continuous-time
controllers cannot be united using a continuous-time supervisor). 
This additional requirement leads to a control scheme with 
mixed discrete/continuous dynamics, see \cite{TeelKapoor97CDC},
\cite{Prieur01}, and \cite{Efimov.06.Automatica.Uniting}, where
controllers to piece together two given state-feedback laws are proposed.
Based on these techniques,
different applications have been considered, such as the stabilization of the inverted pendulum \cite{SanfeliceTeelGoebelPrieur06ACC} and the position and orientation of a mobile robot \cite{Sanfelice.ea.08.CDC.Supervisor}.
These ideas have been extended in
\cite{SanfeliceTeel07ACC} to allow for the combination of multi-objective
controllers, including state-feedback laws as well as open-loop control laws.
More recently, they have also been extended to the case when, rather than state-feedback, only
output-feedback controllers are available
\cite{PrieurTeel:ieee:11}.
A trajectory-based approach for the design of robust multi-objective controllers
that regulate a particular output to zero while keeping another output within
a prescribed limit was introduced in \cite{Efimov.ea.09}.
In the context of performance, a trajectory-based approach
was also employed in \cite{Efimov.ea.IJRNC.11}
to generate
dwell-time and hysteresis-based 
control strategies that guarantee an input-output stability property characterizing
closed-loop system performance.

In this paper, we study the robust stabilization
of nonlinear systems of the form
\begin{equation}\label{eq:plant}
\mathcal{P} : \quad
\dot{\xi}  = f_{p}(\xi,u_{p})
\qquad
\xi \in \Cp, \ u_p \in \reals^{m_p}
\end{equation}
via the combination of two hybrid controllers
that use only measurements of outputs of the plant.
The motivation of such a problem is 
twofold.  
On the one hand, the
impossibility of robustly stabilizing an equilibrium point (or set)
with smooth or discontinuous 
control laws (see, e.g.,
\cite{Brockett83})
precludes utilizing uniting controllers that combine smooth or discontinuous (non-hybrid)
state-feedback laws.
On the other hand, the typical limitation of measuring all of the
plant variables for state-feedback control demands the use of output-feedback controllers
as well as 
the use of multiple controllers that can be combined in a systematic manner
to accomplish a given task. 
These challenges emerge
in stabilization problems with information and actuation constraints.
For instance, in motion planning of autonomous vehicles 
for navigation in cluttered environments, in addition to 
unavoidable input constraints,
obstacles introduce topological constraints that restrict the sensing range.
In such scenarios, control algorithms may combine information from multiple 
sensors and select the most appropriate control strategy to execute.
Due to the different properties induced by the individual controllers in such 
applications, we refer to the problem studied in this paper
as the problem of {\em uniting two output-feedback hybrid controllers
with different objectives}, where
one of the controllers steers the trajectories to a set (this is the objective of the global controller)
and another controller asymptotically stabilizes a different target set (this is the objective of
the local controller); cf. \cite{Efimov.ea.09}.

\subsection*{Contributions}

We propose a hybrid controller to solve the problem of uniting two output-feedback laws with different objectives.
Figure~\ref{fig:ControlArchitecture} depicts the proposed solution,
which consists of  
supervising the two output hybrid controllers, which are denoted by $\KK_0$ and $\KK_1$,
with ``local" and ``global" stabilizing capabilities, respectively.
By combining a discrete and several continuous states, for any compact set of initial conditions, we design
a robustly stabilizing supervisory algorithm with a basin of attraction containing the given compact set of initial conditions,
i.e., the controller renders a target set
semi-globally 
asymptotically
stable.
The  supervisory algorithm consists of a hybrid controller,
which is denoted by $\KK_s$, and uses logic-based switching to unite controllers $\KK_0$ and $\KK_1$.
Our approach builds from the ideas in \cite{PrieurTeel:ieee:11} on uniting output-feedback
continuous-time controllers and
in 
\cite{Morse.96.TAC.SupervisoryP1,Morse.97.TAC.SupervisoryP2,Hespanha.02.Automatica.Supervisory,Sanfelice.ea.08.CDC.Supervisor}
on supervisory control algorithms.

%\begin{figure}[h]
%  \begin{center}  
%  \psfragfig*[width=.48\textwidth]{Uniting_ControllerArchitectureold}
%  {            
%    \psfrag{reference}[][][0.9]{reference}
%    \psfrag{supervisor}[][][0.9]{\hspace{-0.18in} supervisor}
%    \psfrag{local}[][][0.9]{\hspace{-0.13in} controller}
%    \psfrag{global}[][][0.9]{\hspace{-0.13in} controller}
%    \psfrag{controller1}[][][0.9]{\hspace{-0.28in} $\KK_0$}
%    \psfrag{controller}[][][0.9]{\hspace{-0.22in} $\KK_1$}
%    \psfrag{y}[][][0.9]{\ \ \ $y$}
%    \psfrag{y0}[][][0.9]{\!\!\! $y_{p,0}$}
%    \psfrag{y1}[][][0.9]{\!\!\! $y_{p,1}$}
%    \psfrag{q}[][][0.9]{$q$}
%    \psfrag{v}[][][0.9]{\ \ \ $\KK_s$}
%    \psfrag{P}[][][0.9]{\ \ $\cal P$}
%    \psfrag{K0}[][][0.9]{}
%    \psfrag{K1}[][][0.9]{}
%    \psfrag{Ks}[][][0.9]{\hspace{-0.25in} ($z_0, z_1, \tau$)}
%    \psfrag{plant}[][][0.9]{\!\! plant}
%    \psfrag{u0}[][][0.9]{$\kappa_{c,0}$}
%    \psfrag{u1}[][][0.9]{$\kappa_{c,1}$}
%    \psfrag{u}[][][0.9]{$\kappa_{c}$}
%    }
%    \caption{Proposed control approach for Problem~($\star$).}             
%    \label{fig:ControlArchitecture}
%  \end{center}             
%\end{figure}      
\begin{figure}[h]
  \begin{center}  
  \psfragfig*[width=.48\textwidth]{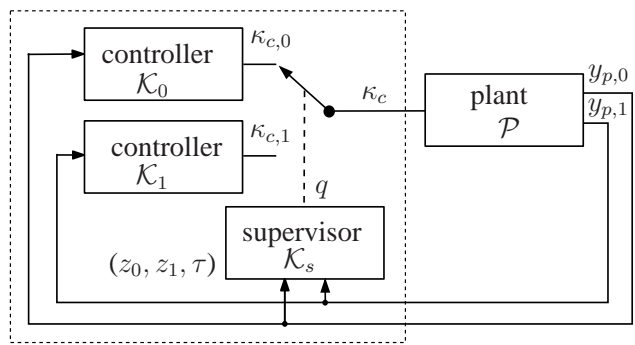}
  {            
    \psfrag{tagsup}[][][0.9]{\hspace{-0.18in} supervisor}
    \psfrag{tag1}[][][0.9]{\hspace{-0.13in} controller}
    \psfrag{global}[][][0.9]{\hspace{-0.13in} controller}
    \psfrag{tagC1}[][][0.9]{\hspace{-0.28in} $\KK_0$}
    \psfrag{tagC}[][][0.9]{\hspace{-0.22in} $\KK_1$}
    \psfrag{y}[][][0.9]{\ \ \ $y$}
    \psfrag{y0}[][][0.9]{\!\!\! $y_{p,0}$}
    \psfrag{y1}[][][0.9]{\!\!\! $y_{p,1}$}
    \psfrag{q}[][][0.9]{$q$}
    \psfrag{v}[][][0.9]{\ \ \ $\KK_s$}
    \psfrag{P}[][][0.9]{\ \ $\cal P$}
    \psfrag{K0}[][][0.9]{}
    \psfrag{K1}[][][0.9]{}
    \psfrag{Ks}[][][0.9]{\hspace{-0.4in} ($z_0, z_1, \tau$)}
    \psfrag{tagpl}[][][0.9]{\!\! plant}
    \psfrag{u0}[][][0.9]{$\kappa_{c,0}$}
    \psfrag{u1}[][][0.9]{$\kappa_{c,1}$}
    \psfrag{u}[][][0.9]{$\kappa_{c}$}
    }
    \caption{Proposed control approach for Problem~($\star$).}             
    \label{fig:ControlArchitecture}
  \end{center}             
\end{figure}      

The features of the proposed hybrid supervisor include:
\begin{itemize}
\item
{\it Uniting of hybrid controllers:} 
controllers $\KK_0$ and $\KK_1$ are not restricted
to being continuous-time controllers; instead, they can be hybrid controllers
involving continuous and discrete variables.
In this way, the proposed solution extends the technique of uniting two continuous-time
controllers available in the literature to the case when the individual controllers
are hybrid, which, in turn, permits applying 
the uniting method to plants that cannot be
robustly stabilized 
by smooth or discontinuous control laws.
\item
{\it Controllers with different objectives:}
controllers $\KK_0$ and $\KK_1$ can have different objectives in the
sense that they may stabilize different attractors. 
This enables the systematic design of controllers that
steer trajectories to a certain point (or set) from where local controllers
can take over and stabilize the desired point (or set).
This procedure has been heuristically used in robotic applications \cite{Burridge.ea.96.ThrowAndCatch}.
\item
{\it Output feedback without underlying input-output-to-state stability assumption on the plant:}
for the solution of the uniting problem of interest (see Problem~($\star$) in Section~\ref{sec:Uniting})
 the proposed hybrid supervisor requires an output-to-state stability
property for each of the closed-loop systems resulting when the 
individual controllers are used.  
This assumption is weaker that the input-output-to-state stability condition on the plant
in  \cite{PrieurTeel:ieee:11}.
The mechanism enabling this relaxation 
is a timer state included in the proposed hybrid supervisor.
\end{itemize}
In this work, each of the output-feedback hybrid controllers is known to
confer certain properties to each of the resulting closed-loop systems:
the first controller renders, for the plant state, a target compact set
locally asymptotically stable, while the second controller renders
a particular compact set attractive.
As a difference to the controllers in 
\cite{SanfeliceTeel07ACC,Efimov.ea.09,Efimov.ea.IJRNC.11},
the individual controllers can be hybrid and their objectives given in terms of compact sets rather than
equilibrium points (the latter feature actually enables the use of hybrid controllers as
these typically stabilize sets larger than a single point; see \cite{Goebel.ea.09.CSM} for a discussion).
Note that as a difference to \cite{Efimov.ea.IJRNC.11}, where switching
times are optimally computed,
the objective of the proposed hybrid supervisor is
to robustly stabilize a desired compact set.
Our construction exploits the fact that, as established in
\cite{Sontag.Wang.97} for continuous-time nonlinear systems
and generalized to hybrid systems in \cite{Cai.Teel.2011.SCL,Cai.Teel.ACC.06},
this property implies the existence of an estimator of the norm of the state.
We work within the hybrid systems framework of
\cite{Goebel.ea.09.CSM} (see also \cite{Hybrid01,GoebelTeel06})
and employ results on robust asymptotic stability
reported in \cite{GoebelTeel06}.
Two examples involving systems
with input constraints and limited information
are used throughout the paper
to illustrate
the application of our results.

\subsection*{Organization of the paper}

The remainder 
of the
 paper is organized as follows.  After basic notation
is introduced, Section~\ref{sec:HybridSystems} presents
a short description of the framework used for analysis.
The main result follows in Section~\ref{sec:Uniting}.
This section 
starts by introducing the problem to be solved,
the proposed formulation of a solution, and the required assumptions.
In addition to presenting a  design procedure for the supervisor,
it establishes a robust stability property of the closed-loop system.
Examples are introduced throughout the paper to illustrate
the ideas.  In Section \ref{sec:examples},  
the proposed hybrid supervisor is applied to the systems
in these examples.

We use the following notation and definitions throughout the paper.
$\reals^{n}$ denotes $n$-dimensional Euclidean space.
$\realsplus$ denotes the nonnegative real numbers, i.e.,
 $\realsplus=[0,\infty)$.
$\nats$ denotes the natural numbers including $0$, i.e.,
 $\nats=\left\{0,1,\ldots \right\}$.
$\ball$ denotes the open unit ball in Euclidean space centered at the origin.
Given a vector $x\in \reals^n$, $|x|$ denotes the Euclidean vector norm.
Given a set $S$, $\ol{S}$ denotes its closure.
Given a set $S\subset \reals^n$ and a point $x \in \reals^n$,
 $|x|_{S}:= \inf_{y \in S} |x-y|$. 
 The notation $F: S\rightrightarrows S$
 indicates that $F$ is a set-valued map that 
 maps points in $S$ to subsets of $S$.
For simplicity in the notation, 
given vectors $x$ and $y$,
we write,
when convenient, 
$[x^\top y^\top]^\top$
with the shorthand notation $(x,y)$.
A function $\alpha : \realsplus \to \realsgeq$ is said to belong
 to the class ${\mathcal K}$ if it is continuous, zero at zero, and
 strictly increasing.
A function $\alpha : \realsplus \to
\realsgeq$ is said to belong to the class ${\mathcal K}_{\infty}$ if it
belongs to the class ${\mathcal K}$ and is unbounded.
A function $\beta: \realsplus \times \realsplus \rightarrow
 \realsplus$ is said to belong to class $\mathcal{K}\mathcal{L}$ if
 it is nondecreasing in its first argument, nonincreasing in its
 second argument, and $\lim_{s \searrow 0} \beta(s,t) =
 \lim_{t \rightarrow \infty} \beta(s,t) = 0$.
A function $\beta:\reals_{\geq 0} \times \reals_{\geq 0} \times
 \reals_{\geq 0} \rightarrow \reals_{\geq 0}$ is said to belong to
 class $\mathcal{K}\mathcal{L}\mathcal{L}$ if, for each $r \in
 \reals_{\geq 0}$, the functions $\beta(\cdot,\cdot,r)$ and
 $\beta(\cdot,r,\cdot)$ belong to class $\mathcal{K}\mathcal{L}$.

\section{Hybrid Systems Preliminaries}
\label{sec:HybridSystems}

In this paper, we consider hybrid systems as in
\cite{Goebel.ea.09.CSM,Goebel.ea.11},
where solutions can evolve
continuously (flow) and/or discretely (jump) depending on the
continuous and discrete dynamics  
of the hybrid systems,
and the sets where those dynamics
apply. In general, a hybrid system $\HS$ is given by data $(h,C,F,D,G)$
and can be written in the compact form
$$
\HS : \quad
\left\{
\begin{array}{lcl}
\dot{\chi} &\in& F(\chi) \qquad \chi\in C \\
 \nonumber \chi^+ &\in& G(\chi) \qquad \chi\in D\\
 y &=& h(\chi),
\end{array}
\right.
$$
where $\chi$ is the state taking values from $\reals^n$,
the set-valued map $F$ defines the continuous dynamics on the
set $C$ and the set-valued map $G$ defines the discrete dynamics on
the set $D$. The notation $\chi^+$ indicates the value of the state
$\chi$ after a jump\footnote{Precisely, 
$\chi^+ = \chi(t,j+1)$.}.
The function
$h$ defines the output.
Solutions to $\HS$ will be given on {\em hybrid time domains}, 
which are subsets $\SSS$ of $\realsgeq \times \nats$
that, for every
$(T,J)\in \SSS$, $\SSS\ \cap\ \left( [0,T]\times\{0,1,\ldots J\}
\right)$ can be written as
$\bigcup_{j=0}^{J-1} \left([t_j,t_{j+1}],j\right)$
for some finite
sequence of times $0=t_0\leq t_1 \ldots \leq t_J$.  
A solution 
to $\HS$ will consist of a hybrid time domain $\dom \chi$
and a {\em hybrid arc} $\chi:\dom \chi\to \reals^n$, 
which is a function with the property that $\chi(t,j)$ is locally absolutely continuous 
on $I_j := \defset{t}{(t,j)\in\dom\chi}$ for each $j\in \nats$,
satisfying the dynamics imposed by $\HS$.
More precisely,
the following hold:
\begin{itemize}
\item[]{(S1)} For each $j\in\nats$ such that $I_j$ has nonempty interior
\begin{equation}
\label{S1}
\begin{array}{l}
\chi(t,j) \in C\quad \mbox{ for all } t \in [\min I_j,\sup I_j) \\
\dot{\chi}(t,j)\in F(\chi(t,j)) \quad \mbox{ for almost all } t \in I_j;
\end{array}
\end{equation}
\item[]{(S2)} For each $(t,j)\in \dom \chi$ such that $(t,j+1)\in \dom \chi$,
\begin{equation}
\label{S2}
\chi(t,j) \in D, \ \ \ \chi(t,j+1)\in G(\chi(t,j)).
\end{equation}
\end{itemize}
Hence, 
solutions are parameterized by $(t,j)$, where
$t$ is the ordinary time and $j$ 
corresponds to the number of jumps.
A solution $\chi$ to $\HS$ is said to be {\it complete} if $\dom \chi$ is
unbounded, {\it Zeno} if it is complete but the projection of $\dom
\chi$ onto $\realsgeq$ is bounded, and {\it maximal} if there does not
exist another hybrid arc $\chi'$ such that $\chi$ is a truncation of
$\chi'$ to some proper subset of $\dom \chi'$.  For more details about
this hybrid systems framework, we refer the reader to
\cite{Goebel.ea.09.CSM}.

When the data $(h,C,F,D,G)$ of $\HS$
satisfies the conditions given next,
hybrid systems are well posed in the sense that they inherit several good structural properties
of their solution sets. These
include sequential compactness of the solution set, closedness of perturbed and
unperturbed solutions, etc.
We refer the reader to \cite{GoebelTeel06} (see also \cite{Hybrid01})
and \cite{SanfeliceGoebelTeel05SIAM}
for details on and consequences of these conditions.

\begin{definition}{\bf (Well-posed hybrid systems)}
\label{def:ba}
The hybrid system $\HS$ with data $(h,C,F,D,G)$
is said to be {\em well posed} if it satisfies the following {\em hybrid basic conditions}:
 the sets $C$ and $D$ are closed,
the mappings $F:C \rightrightarrows \reals^n$ and $G:D\rightrightarrows \reals^n$ are outer semicontinuous and locally bounded,\footnote{A set-valued mapping $G$ defined on $\reals^n$ is {\it outer
 semicontinuous} if for each sequence $x_{i} \in \reals^n$ converging to a
point $x \in \reals^n$ and each sequence $y_{i} \in G(x_{i})$ converging to a
point $y$, it holds that $y \in G(x)$. It is {\it locally bounded} if,
for each compact set $\K \subset \reals^n$ there exists $\mu>0$ such that
$
\cup_{x \in \K} G(x) \subset \mu \ball$.
}
$F(x)$ is nonempty and convex for all $x\in C$, $G(x)$ is nonempty for
all $x\in D$,
and 
$h:\reals^n\to \reals^m$ 
is continuous.
\end{definition}

\section{Uniting Two Output-Feedback Hybrid Controllers Using a Hybrid Supervisor}
\label{sec:Uniting}

\subsection{Problem statement, solution approach, and assumptions}

We consider the stabilization of a compact set for
nonlinear control systems of the form \eqref{eq:plant}
with only measurements of two outputs $y_{p,0}$ and $y_{p,1}$ given by functions of the state $h_0$ and $h_1$,
respectively, 
where $f_p$ is a continuous function.
That is, we are interested in solving the following problem:
\begin{itemize}
\item[($\star$)]  
Given compact sets $\A_0,\K_0 \subset\reals^{n_p}$ 
and continuous functions $h_0,h_1$ defining outputs $y_{p,0}=  h_0(\xi)$ 
and $y_{p,1}= h_1(\xi)$ of \eqref{eq:plant},
design an output feedback controller
that renders $\A_0$ asymptotically stable 
with a basin of attraction containing $\K_0$.\footnote{It is desired that 
the basin
of attraction 
of the closed-loop system
contains $\K_0$ when 
projected onto $\reals^{n_p}$.}
\end{itemize}

As shown
 in Figure~\ref{fig:ControlArchitecture},
the proposed approach to solve this problem consists of 
supervising two output hybrid controllers, which are denoted by $\KK_0$ and $\KK_1$,
with ``local" and ``global" stabilizing capabilities, respectively, 
which are properties that will be made precise below.
The  supervisory algorithm consists of a hybrid controller,
which is denoted by $\KK_s$, that uses logic-based decision
making to unite controllers $\KK_0$ and $\KK_1$.
The individual controllers $\KK_0$ and $\KK_1$ have state $\zeta_0$ and
$\zeta_1$, both in $\reals^{n_c}$, respectively.\footnote{The case
where the hybrid controllers have a dynamical state $\zeta_0$
(respectively, $\zeta_1$) in a set $\reals^{n_{c0}}$ (respectively,
$\reals^{n_{c1}}$) of different dimension $n_{c0}\neq n_{c1}$ 
can be treated similarly by embedding both sets into the set of larger
dimension.}
For each $i \in \{ 0,1 \}$, the hybrid controller
$\KK_i=(\kappa_{c,i},C_{c,i},f_{c,i},D_{c,i},g_{c,i})$
is given by
\begin{equation}
\label{eq:controller}
\KK_i: \quad \left\{
   \begin{array}{llll}
      \dot{\zeta}_{i} & =  & f_{c,i}(u_{c,i},\zeta_i)
 &  \ (u_{c,i},\zeta_i) \in C_{c,i} \\
      \zeta_i^{+} & \in &  g_{c,i}(u_{c,i},\zeta_i) & \ (u_{c,i},\zeta_i) \in D_{c,i}\\
              y_{c,i}& = & \kappa_{c,i}(u_{c,i},\zeta_i),
   \end{array}
\right.
\end{equation}
where 
$\zeta_i \in \Re^{n_{c}}$ is the $i$-th controller's state,
$u_{c,i} \in \Re^{m_{c,i}}$ the $i$-th controller's input,
$C_{c,i}$ and $D_{c,i}$ are subsets of $\reals^{m_{c,i}}\times\reals^{n_c}$,
$\kappa_{c,i}:\reals^{n_c}\to\reals^{m_p}$ is the $i$-th controller's output,
$f_{c,i}:C_{c,i} \rightarrow \Re^{n_{c}}$, and
$g_{c,i}:D_{c,i} \rightrightarrows \Re^{n_{c}}$.
For each $i \in \{ 0,1 \}$, the $i$-th controller $\KK_i$ measures the plant's output $y_{p,i} =
h_i(\xi)$ only and,
via the
assignment $u_{c,i}=y_{p,i}$, $u_{p}=y_{c,i}$
defines 
the hybrid closed-loop system
denoted by $({\cal P},\KK_i)=(h_i,C_i,f_i,D_i,g_i)$
with state 
$(\xi, \zeta_i)\in \reals^n$,
$n = n_p + n_c$,
and given by
\begin{equation}
  \begin{array}{lcll}
  \renewcommand{\arraystretch}{1.1}
 \ \matt{\dot{\xi}\\ \dot{\zeta}_i}  =  f_i(\xi,\zeta_i):=
 \matt{f_{p}(\xi,\kappa_{c,i}(h_i(\xi),\zeta_i))  \\ f_{c,i}(h_i(\xi),\zeta_i)}
    & (\xi,\zeta_i)\! \in\! C_i,
\\
\\
  \renewcommand{\arraystretch}{1.5}
\matt{ \xi^{+} \\ \zeta_i^{+}}  \in  g_i(\xi,\zeta_i) :=
\matt{\xi \\ g_{c,i}(h_i(\xi),\zeta_i)}
& (\xi,\zeta_i)\! \in\! D_i, \\
y_i =  h_i(\xi),
\end{array}
\label{eq:closed-loop}
\end{equation}
where $y_i$ is the output, 
$$\begin{array}{c}
C_i:=\defset{(\xi,\zeta_i)}{\xi\in \Cp,(h_i(\xi),\zeta_i) \in C_{c,i}},\\
D_i := \defset{(\xi,\zeta_i)}{\xi\in \Cp,(h_i(\xi),\zeta_i) \in D_{c,i}}.\end{array}$$
(An assignment different from $u_{c,i}=y_{p,i}$, $u_{p}=y_{c,i}$ will be 
employed when a hybrid supervisor is used -- see Theorem \ref{thm:NominalAS}.) 
We say that the controller $\KK_i$ is well posed when the resulting
closed-loop system from controlling the plant \eqref{eq:plant} with continuous right-hand side 
is well posed
as in Definition~\ref{def:ba}.

The controllers $\KK_i$ are assumed to induce the properties that, for $i=0$, a compact set
$\A_0 \times \C_0 \subset \reals^n$, where $\C_0 \subset \reals^{n_c}$, is
locally asymptotically stable for $({\cal P},\KK_0)$ and, for $i=1$, a compact
set $\A_1 \times \C_1 \subset \reals^n$, $\C_1 \subset
\reals^{n_c}$, is attractive for $({\cal P},\KK_1)$.
For a combination of both controllers to work, 
the set $\A_1$ will have to be contained in the plant component of the basin of attraction of $\KK_0$.
In such a case, the said properties of $\KK_0$ and $\KK_1$ readily suggest that, 
when far away from $\A_0$, $\KK_1$ can be used to steer the plant's state to a region from where
$\KK_0$ can be used to
asymptotically
 stabilize $\A_0$.
However,
these controllers cannot be combined using
supervisory control techniques in the literature (see, e.g., \cite{Sanfelice.ea.08.CDC.Supervisor} and the references therein)
due to being hybrid and to the lack of full measurements of $\xi$.
We resolve this issue by designing two norm observers.
The existence of such observers is guaranteed when the hybrid controllers induce an output-to-state stability (OSS) property.
More precisely, this OSS property 
assures
 the 
existence of an (smooth)
exponential-decay OSS-Lyapunov function $V_i$ with respect
to $\A_i\times\C_i$ for $(\P,\KK_i)$; see \cite[Theorem 3.1]{Cai.Teel.2011.SCL}.
As defined in \cite[Definition 2.2]{Cai.Teel.2011.SCL},
$V_i:\reals^n\to\realsgeq$ is
such that there exist class-${\cal K}_{\infty}$ functions $\alpha_{i,1}, \alpha_{i,2}$,
class-$\cal K$ function $\gamma_{i}$,
and $\eps_i \in (0,1]$ satisfying:
for all $(\xi,\zeta_i) \in \reals^n$,
\begin{equation}
\begin{array}{rcl}
\alpha_{i,1}(|(\xi,\zeta_i)|_{\A_i\times\C_i}) \leq V_i(\xi,\zeta_i) \leq \alpha_{i,2}(|(\xi,\zeta_i)|_{\A_i\times\C_i});
\label{eqn:Vibound}
\end{array}\end{equation}
for all  $(\xi,\zeta_i) \in C_i$,
\begin{equation}
\begin{array}{rcl}
\langle \nabla V_i(\xi,\zeta_i), f_i(\xi,\zeta_i) \rangle \leq - \eps_i V_i(\xi,\zeta_i) + \gamma_{i}(|h_i(\xi)|);
\label{eqn:VidotBound}
\end{array}\end{equation}
for all $(\xi,\zeta_i) \in D_i$,
\begin{equation}
\begin{array}{rcl}
\max_{g \in g_i(\xi,\zeta_i)}V_i(g) - V_i(\xi,\zeta_i) \leq - \eps_i V_i(\xi,\zeta_i) + \gamma_{i}(|h_i(\xi)|).
\label{eqn:ViChangeBound}
\end{array}\end{equation}
The next assumption guarantees that the resulting
closed-loop systems $({\cal P},\KK_0)$ and $({\cal P},\KK_1)$ 
satisfy these properties.

\begin{assumption}
\label{assumption:1}
Given a compact set $\A_0 \subset \reals^{n_p}$ and
continuous functions 
$f_p:\reals^{n_p}\times\reals^{m_p}\to\reals^{n_p}$,
$h_0:\reals^{n_p}\to\reals^{m_{c,0}}$, 
$h_1:\reals^{n_p}\to\reals^{m_{c,1}}$, 
where $h_0(\xi) = 0$ for all  $\xi \in \A_0$,
assume there exist 
compact sets 
$\A_1 \subset \reals^{n_p}$, $\C_0, \C_1 \subset \reals^{n_c}$, where
$h_1(\xi) = 0$ for all 
$\xi \in \A_1$,
such that:
\begin{enumerate}
\item
A well-posed hybrid controller
$\KK_0=$\\$(\kappa_{c,0},C_{c,0},f_{c,0},D_{c,0},g_{c,0})$
for the plant output $y_{p,0} = h_0(\xi)$ 
inducing the following properties exists: 
\begin{enumerate}
\item
Stability: For each $\eps >0$  there exists $\delta >0$ such that every solution $(\xi,\zeta_0)$
to $(\P,\KK_0)$ with \
$|(\xi(0,0),\zeta_0(0,0))|_{\A_0 \times \C_0} \leq \delta$ satisfies $|(\xi(t,j),\zeta_0(t,j))|_{\A_0 \times {\C}_0} \leq \eps$ for all $(t,j) \in \dom (\xi,\zeta_0)$;\footnote{The plant state $\xi$ is parameterized by $(t,j)$ since
it is a 
state
component of the hybrid system $({\cal P},\KK_0)$, 
whose solutions
are defined on hybrid time domains.}
\item
Attractivity: There exists $\mu >0$ such that every solution $(\xi,\zeta_0)$
to $(\P,\KK_0)$ with \\
$|(\xi(0,0),\zeta_0(0,0))|_{\A_0 \times {\C}_0} \leq \mu$ is complete and satisfies 
$$\lim_{t+j \to \infty} |(\xi(t,j),\zeta_0(t,j))|_{\A_0 \times {\C}_0} = 0;$$
\item
Output-to-state stability 
(OSS): The hybrid system $({\cal P},\KK_0)$ with output $y_{p,0}=h_0(\xi)$ 
is output-to-state stable with respect to $\A_0\times\C_0$. 
Let $V_0$ denote an OSS-Lyapunov function 
associated with this property, and let $\gamma_0 \in \classK$ and $\eps_0>0$
satisfy \eqref{eqn:VidotBound} and \eqref{eqn:ViChangeBound} with $i=0$.
Let
$\eps_{0,b}>0$ define an estimation of the 
basin of attraction ${\cal B}_0$ of $(\P,\KK_0)$
of the form
$\defset{(\xi,\zeta_0)}{V_0(\xi,\zeta_0) \leq \eps_{0,b}}$.
\end{enumerate}
\item
A well-posed hybrid controller $\KK_1=$\\$(\kappa_{c,1},C_{c,1},f_{c,1},D_{c,1},g_{c,1})$
for the plant output $y_{p,1} = h_1(\xi)$
inducing the following properties exists:
\begin{enumerate}
\item Attractivity: Every maximal solution $(\xi,\zeta_1)$
to $(\P,\KK_1)$ 
is complete and satisfies
$$\lim_{t+j \to \infty}|(\xi(t,j),\zeta_1(t,j))|_{\A_1 \times {\C}_1}=0;$$
\item Output-to-state stability: The hybrid system $({\cal P},\KK_1)$ with output $y_{p,1}=h_1(\xi)$ 
is output-to-state stable with respect to $\A_1\times\C_1$. 
Let $V_1$ denote an OSS-Lyapunov function associated with this property.
\end{enumerate}
\item
There exist  $\eps_{0,a},\eps_{1,b}>0$ such that
$\eps_{0,a} < \eps_{0,b}$ and,
 for each solution $(\xi,\zeta_0)$ to $(\P,\KK_0)$
from
$$
\defset{\xi \in \reals^{n_p}}{V_1(\xi,\zeta_1) \leq \eps_{1,b},\ \zeta_1 \in \C_1} \times
\C_0,
$$
we have
\begin{equation}\label{eqn:V1set}
\gamma_0(|h_0(\xi(t,j))|) < \eps_{0,a}\eps_0  \ \ \forall (t,j)\in\dom (\xi,\zeta_0).
\end{equation}
\end{enumerate}
\end{assumption}

\begin{remark}\label{remark:Assumptions}
Assumption~\ref{assumption:1} assures the existence of individual controllers
with enough properties so that the uniting problem of interest 
is at all solvable and the proposed approach provides a solution to it.
More precisely, 
items 1.a and 1.b are
required so that the local stability 
requirement in Problem ($\star$) is attainable
while
2.a is needed so that the semi-global stability 
requirement therein can be met.
The other assumptions are particular to our proposed
solution.
Items 
 1.c 
  and 2.b are imposed so that norm observers  
can be constructed.
Item 3
 permits the combination of the two
controllers using a hybrid supervisor
by ensuring that the compact set $\A_1$, which is the plant component of the 
set rendered attractive with the controller $\KK_1$,
is included in the basin of
attraction of the closed-loop system with the controller $\KK_0$. 
In this way, $\A_0 \times \Phi_0$ can be asymptotically stabilized
once $\KK_1$ steers the plant state nearby $\A_1$.
Note that 
items 1.a, 1.b,
and 2.a
are the hybrid version of the assumptions in \cite{PrieurTeel:ieee:11}.
Items 
1.c
and 2.b relax the assumptions in \cite{PrieurTeel:ieee:11}
as rather than asking for 
input-output-to-state stability (IOSS) of the plant, 
they impose
OSS properties of the  closed-loop systems $({\cal P},\KK_0)$ and $({\cal P},\KK_1)$.
\end{remark}

The stabilizing property induced by controller $\KK_0$
in Assumption~\ref{assumption:1} holds when 
the nonlinear system is locally stabilizable to the set $\A_0$ by hybrid feedback.
Note that hybrid feedback permits stabilizing a larger class
of systems than standard continuous feedback.
Examples of systems that can be asymptotically stabilized
by hybrid feedback include the nonholonomic integrator and
Artstein circles \cite{PrieurGoebelTeel07},
the pendubot \cite{SanfeliceTeel07ACC},
and rigid bodies \cite{Mayhew.ea.10.TAC.a}.
The attractivity property induced by the controller $\KK_1$
in Assumption~\ref{assumption:1} holds when the trajectories
of the plant can be asymptotically steered to the set $\A_1$
(contained in the basin of attraction of the local controller).
Note that, as a difference to controller $\KK_0$, 
it is not required for controller $\KK_1$ to render the said set
stable.  This feature of the proposed controller
allows for the design of $\KK_0$ and $\KK_1$ separately, 
being
item 3 of Assumption~\ref{assumption:1} 
a common design 
constraint.

Next, we introduce an example and associated control problem
for which the supervision of 
two controllers with properties as in Assumption~\ref{assumption:1}
will be applied.

\begin{example}
\label{ex:Avoidance1}
Consider the stabilization of the point $\{\xi^*\}$
for the point-mass system
$\dot{\xi} =
u_p$,
where
$\xi\in\reals^2$ is the state and $u_{p} = [u_1\ u_2]^{\top}$ is the control input.\footnote{See \cite{Efimov.ea.09} where the problem of stabilizing a unicycle while ensuring obstacle avoidance is studied.}
A controller is to be designed to solve the following control problem: 
guarantee that the solutions to the plant
avoid a neighborhood around the point $\overline{\xi}$, 
which is
given by
$\cal N = \overline{\xi} +\hat{\alpha}\ball$
and represents an obstacle,
and that converge to the target point $\xi^*$.
Convergence to the target can be attained
by steering the solutions in the clockwise or in the counter-clockwise
direction around the obstacle, depending on the initial condition.
Measurements of the distance to the target may not be available
from points where the target is not visible
due to the presence of the obstacle.
Due to the topological constraint of the stabilization
task and the limited measurements,
a single controller or a controller uniting
two controllers with the same objectives 
would be difficult to design.

To solve the stated control problem, 
functions defining potential fields 
capturing the presence
of the obstacle and
vanishing at some point $\xi^\circ$ from where the target is visible, i.e.,
from points where there is a ``line-of-sight'' to the target point $\xi^*$, can be generated.
Then,
a gradient descent controller can be designed to steer the 
state of the point-mass system
to nearby the intermediate point $\xi^{\circ}$.
In this way, the point $\xi^{\circ}$ would define the set $\A_1$ and the gradient-descent 
controller would define $\KK_1$. 
This controller would use measurements of the functions defining the 
potential fields as well as their gradients.  These functions would define the plant's output $y_{p,1}$.
A particular construction 
of a hybrid controller implementing a robust gradient-descent-like strategy
and 
satisfying the conditions in Assumption~\ref{assumption:1}.2
is given in Section~\ref{ex:2}.
To satisfy the conditions in Assumption~\ref{assumption:1}.1,
a ``local'' controller capable of asymptotically stabilizing $\xi^*$
from nearby $\xi^\circ$ would play the
role of the controller $\KK_0$ above, with $\A_0$ given by $\{\xi^*\}$.
Due to $\xi^\circ$ being at a location unobstructed by the obstacle,
this controller could use relative position measurements
to the target, which would define the plant's output $y_{p,0}$.
Item 3 of Assumption~\ref{assumption:1} will be satisfied
by placing $\A_1$ in the basin of attraction induced by $\KK_0$.
\hfill $\triangle$
\end{example}

As pointed out in Remark~\ref{remark:Assumptions},
items 
1.c and 2.b 
in 
Assumption~\ref{assumption:1}
assure
OSS the existence of exponential-decay OSS-Lyapunov functions with respect
to $\A_i\times\C_i$ for $(\P,\KK_i)$.
As stated in \cite[Proposition 2]{Cai.Teel.ACC.06}, a norm estimator for the state $(\xi,\zeta_i)$ (and, hence, for $\xi$) exists. A particular construction is
\begin{equation}
\begin{array}{rcl}
\label{eqn:NormEstimator}
\begin{array}{llll}
\dot{z}_i & = & -\eps_i z_i + \gamma_i(|h_i(\xi)|) & \qquad (\xi,\zeta_i) \in C_i, \\
z_i^+ & = & (1-\eps_i) z_i + \gamma_i(|h_i(\xi)|) & \qquad (\xi,\zeta_i) \in D_i.
\end{array}
\end{array}
\end{equation}
In fact, given a solution $(\xi,\zeta_i)$ to $(\P,\KK_i)$, using \eqref{eqn:VidotBound} and \eqref{eqn:ViChangeBound},
for each $j \in \nats$ and for almost all $t \in I_j$, $I_j$ with nonempty interior, $(t,j) \in \dom (\xi,\zeta_i)$,
we have$\\[1em]
\frac{d}{dt} \left(V_i(\xi(t,j),\zeta_i(t,j)) - z_i(t,j) \right)\!\! \leq$\hfill\null \\[0.7em]\null\hfill$ - \eps_i \left(V_i(\xi(t,j),\zeta_i(t,j)) - z_i(t,j)\right),
\quad$\\[1em]
and, for each $(t,j) \in \dom  (\xi,\zeta_i)$ such that $(t,j+1) \in \dom  (\xi,\zeta_i)$, we have
$\\[1em]
V_i(\xi(t,j+1),\zeta_i(t,j+1)) - z_i(t,j+1) \leq$\hfill\null \\[0.7em]\null\hfill$ (1 - \eps_i) (V_i(\xi(t,j),\zeta_i(t,j)) - z_i(t,j)).
\quad$\\[1em]
Using the upper bound in \eqref{eqn:Vibound}, it follows that, for all $(t,j)\in\dom (\xi,\zeta_i)$,
$V_i(\xi(t,j),\zeta_i(t,j)) \leq  z_i(t,j) + \exp(-\eps_i t) (1 - \eps_i)^j  \left(V_i(\xi(0,0),\zeta_i(0,0)) - z_i(0,0)\right)
\leq 
z_i(t,j) + \exp(-\eps_i t) (1 - \eps_i)^j \left(\alpha_{i,2}(|(\xi(0,0),\zeta_i(0,0))|_{\A_i\times\C_i})\right.$\\$\left. - z_i(0,0)\right).
$
Assuming, without loss of generality, that $\alpha_{i,2}(s) \geq s$ for all $s \geq 0$ and defining
$
\beta_i(s,t,j) := 2\exp(-\eps_i t) (1 - \eps_i)^j \alpha_{i,2}(s)
$
gives
for any solution $(\xi,\zeta_i)$ to $(\P,\KK_i)$
\\[0.7em]$\null\quad
V_i(\xi(t,j),\zeta_i(t,j)) \leq$\refstepcounter{equation}\label{eqn:VdecreaseEstimator}\hfill$(\theequation)$\\[0.5em]\null\hfill$
z_i(t,j) + \beta_i(|(\xi(0,0),\zeta_i(0,0))|_{\A_i\times\C_i} + |z_i(0,0)|,t,j).
\quad$\\[0.7em]
The following bound on $|(\xi(t,j),\zeta_i(t,j))|_{\A_i\times\C_i}$ follows with \eqref{eqn:Vibound}:
\\[0.7em]$\null\quad
|(\xi(t,j),\zeta_i(t,j))|_{\A_i\times \C_i} \leq$\refstepcounter{equation}\label{eqn:KLLIOSSboundiOnX}\hfill$(\theequation)$\\[0.5em]\null\hfill$
 \alpha_{i,1}^{-1}\left( z_i(t,j) + \beta_i(|(\xi(0,0),\zeta_i(0,0))|_{\A_i\times\C_i} + |z_i(0,0)|,t,j)
\right)
\quad$\\[0.7em]
for all $(t,j)\in\dom (\xi,\zeta_i)$.

The following example illustrates the construction of a norm observer
for a nonlinear system.
This observer will be used in the design of a hybrid supervisor in Section~\ref{ex:1}.

\begin{example}
\label{ex:LimitedInformation1}
Consider the nonlinear system 
\begin{equation}
\begin{array}{rcl}
\label{eqn:plantEx1}
\dot{\xi} = f_p(\xi,u_p)
:=
\left[
\begin{array}{lll}
-\xi_1+(u_{1}- \xi_2)\xi^2_1\\
-\xi_2+\xi^2 _1+ \overline{\alpha} + u_{2}
\end{array}
\right],
\end{array}\end{equation}
where
$\xi\in\reals^2$ is the state and
$u_{p} = [u_1\ u_2]^{\top}$ is the control input.
An 
output-feedback 
controller has been designed for this system in
\cite{AndrieuPraly07}.\footnote{For the case $\overline{\alpha}=0$, 
dynamic output feedback laws for outputs
given by $\xi_1$ or $\xi_2$
that globally asymptotically stabilize the origin in $\reals^2$ have been proposed in \cite{AndrieuPraly07}.}
Measurements of $\xi_1$ and $\xi_2$ are available but not simultaneously.
Consider a  controller $\KK_0$ given by a static feedback controller
that measures $h_0(\xi) := \xi_1$ to stabilize $\xi$ to $\A_0 = \{(0,0)\}$.
Following \eqref{eq:controller},
an example of such a controller is defined by $n_c=0$, $\kappa_{c,0}(\xi):=[0,\ -\overline{\alpha}]^{\top}$,
and no dynamical state (i.e., $C_{c,0} = D_{c,0} = \emptyset$ and $f_{c,0}, g_{c,0}$ are
arbitrary).
For $V_0(\xi) = \frac{1}{2} \xi^\top \xi$, it follows that, for all $\xi \in \reals^2,$\footnote{Using Young's inequality 
to obtain 
$\xi_1^3 \xi_2 \leq \xi_1^6 + \frac{1}{4} \xi_2^2$ 
and
$\xi_1^2 \xi_2\leq \xi_1^4 + \frac{1}{4} \xi_2^2$.}
\begin{eqnarray}
\nonumber\langle \nabla V_0(\xi),f_p(\xi,\kappa_{c,0}(\xi)) \rangle&=&-\xi_1^2  -\xi_1^3\xi_2 - \xi_2 ^2 +\xi_1 ^2\xi_2 
\\
&\leq&  - V_0(\xi)+\xi_1^4 (1 + \xi_1^2).
\label{eqn:V0dot}
\end{eqnarray}
Then, a norm observer for $|\xi|_{\A_0}$ is given by
$
\dot{z}_0 = -z_0 + \gamma_0(|h_0(\xi)|)
$
with
$\gamma_0(s) = s^4(1 +s^2)$ for all $s\geq 0$.
This norm estimator and the controller $\KK_0$ above are such that
Assumption~\ref{assumption:1}.1 holds.
\hfill $\triangle$
\end{example}

In the next section, we provide a solution to Problem ($\star$)
that consists of a hybrid 
supervisor
 coordinating, using control logic and norm observers, the two (well-posed) output-feedback hybrid controllers $\KK_0$ and $\KK_1$.

\subsection{Proposed Control Strategy}
\label{sec:ControlDesign}

As depicted in Figure~\ref{fig:ControlArchitecture},
we propose a hybrid controller $\KK_s$ to supervise $\KK_0$ and $\KK_1$.
This hybrid controller, referred to as the {\em hybrid supervisor}, is
designed to perform the uniting task as follows:
\begin{enumerate}
\item[{\bf A)}] Apply the hybrid controller $\KK_1$ when the estimate of $|\xi|_{\A_1}$ is away from the origin.
\item[{\bf B)}] Permit estimate of $|\xi|_{\A_1}$ to converge.
\item[{\bf C)}] Apply $\KK_0$ when the estimate of $|\xi|_{\A_1}$ is  close enough to zero.
\end{enumerate}
To accomplish these tasks, the supervisor has a discrete state
$q \in \Q := \{0,1\}$
and a timer state $\tau \in \reals$ with 
reset threshold
$\tau^* > 0$.   
The constant $\tau^*$ is a design parameter of the hybrid supervisor.
The dynamics of the state $q$ are designed to indicate that the controller $\KK_q$ is connected to the plant.
While the accomplishment of tasks {\bf A)-C)} with the proposed hybrid supervisor requires finitely many jumps in the state $q$,
the number of jumps in $q$ depends on the initial conditions as well as on the dynamics of the closed-loop system.
A hybrid supervisor implementing tasks {\bf A)-C)} is presented next.

\subsubsection{Supervision of Controller $\KK_1$ ($q=1$)}
\label{sec:GlobalController}

Item 2.a
of
Assumption~\ref{assumption:1} implies that for every solution $(\xi,\zeta_1)$
to $(\P,\KK_1)$ we have
$$
\lim_{t+j \to \infty}\gamma_1(|h_1(\xi(t,j))|) = 0.
$$
Using \eqref{eqn:NormEstimator} for $i=1$,
it follows that $z_1$ also approaches zero, and that, eventually, when $t$ or $j$ are large enough, 
 $|\xi|_{\A_1}$ is small enough.
This suggests that the supervisor should apply $\KK_1$ until, eventually, $z_1$ is small enough.
This can be implemented as follows:
\begin{itemize}
\item Flow according to
\begin{equation}
\begin{array}{c}
\dot{\xi}  =   f_p(\xi,\kappa_{c,1}(h_1(\xi),\zeta_1)), \;
\dot{\zeta}_0  =  0, \;
\dot{\zeta}_1  =  f_{c,1}(h_1(\xi),\zeta_1),\\
\dot{z}_0  = 0,\;
\dot{z}_1  =  -\eps_1 z_1 + \gamma_1(|h_1(\xi)|),\;
\dot{q}  =  0,\;
\dot{\tau}  =  1
\end{array}\label{eqn:FlowsAt1a}
\end{equation}
when, for a design parameter $\eps_{1,a}>0$,
either one of the following conditions hold:
\begin{equation}
\label{eqn:FlowSet1a}
(\xi,\zeta_1) \in C_1
, \ \ 
\zeta_0 \in \C_0,\ \ 
z_0 = 0,\ \  
z_1 \geq \eps_{1,a},\ \
q = 1,
\end{equation}
or
\begin{equation}
\label{eqn:FlowSet1b}
(\xi,\zeta_1) \in C_1, \ \ 
\zeta_0 \in \C_0, \ \ 
z_0 = 0,\ \ 
z_1 \geq 0, \ \ 
q = 1, \ \
\tau \leq \tau^*.
\end{equation}
\item Jump according to
\begin{equation}\label{eqn:UpdateLaw1To0a}
\begin{array}{c}
\xi^+ = \xi,\ \
\zeta_0^+ \in \Phi_0, \ \ \zeta_1^+ \in \Phi_1, \\
z_0^+ = 0, \ \
z_1^+ = 0, \ \
 q^+ = 0,
\ \ \tau^+ = 0
\end{array}\end{equation}
when
\begin{equation}
\label{eqn:JumpSet1To0a}
\zeta_0 \in \C_0, \ \
z_0 = 0, \ \  \eps_{1,a}\geq z_1 \geq 0,  \ \  q = 1, \ \
\tau \geq \tau^*.
\end{equation}
\end{itemize}

The flows defined in \eqref{eqn:FlowsAt1a} enforce, in particular, that $q$ remains constant
and that the estimate of $|\xi|_{\A_1}$ converges.
Condition \eqref{eqn:FlowSet1a} allows flows when the estimate of $|\xi|_{\A_1}$ is not small enough,
while, when condition \eqref{eqn:JumpSet1To0a} holds,
the state $q$ is set to $0$ so that $\KK_0$ is applied.
The state $\zeta_0$ is updated 
to any value in $\C_0$
and the estimator state $z_0$ is reset to zero.
These selections are to properly initialize $\KK_0$.
 However, to guarantee
that the state $\zeta_1$ converges to $\Phi_1$, 
the state is reset to any point in $\Phi_1$.

Due to the impossibility of measuring $\xi$, it is not possible
to ensure that $\xi$ is such that $(\xi,\zeta_0)$ is in the basin of attraction ${\cal B}_0$
after jumps from $q=1$ to $q=0$ occur.
Hence, 
it could be the case that there are jumps from $q=0$ back to $q=1$.
The logic in
\eqref{eqn:FlowsAt1a}-\eqref{eqn:JumpSet1To0a}
uses the timer $\tau$ to guarantee
convergence of the state to ${\cal B}_0$.
The condition 
$\tau  \leq \tau^*$
 in
\eqref{eqn:FlowSet1b} 
allows
the estimate $|\xi|_{\A_1}$  to converge
by enforcing
 that, perhaps after a few jumps to $q=0$ and back to $q=1$,
$\xi$ eventually is so that $(\xi,\zeta_0)$ is in the said basin of attraction.
The conditions involving $z_0$ 
in \eqref{eqn:FlowSet1a}, \eqref{eqn:FlowSet1b}, and
\eqref{eqn:JumpSet1To0a} force $z_0$ to remain 
at zero
along solutions
with $q = 1$.  
These choices facilitate the establishment
of our main result in Section~\ref{sec:MainResult}.
A procedure to design the controller parameters
is given in Section~\ref{sec:DesignProcedure}.

\subsubsection{Supervision of Controller $\KK_0$ $(q=0)$}
\label{sec:LocalController}

From item 1 of Assumption~\ref{assumption:1} and \eqref{eqn:NormEstimator} for $i=0$,
it follows 
that 
$z_0(t,j)$ approaches zero as $\gamma_0(|h_0(\xi(t,j))|)$ approaches zero.
Furthermore,
when $z_0 \leq \eps_{0,a}$, $\zeta_0 \in \C_0$,
and $t$  
or 
$j$ are large enough, it follows from
\eqref{eqn:VdecreaseEstimator} for $i=0$
and 
items 1.b and 3 in Assumption~\ref{assumption:1}
 that after jumps to $q=0$,
$(\xi,\zeta_0)$ will be in the set
\begin{equation}\label{eqn:V0set}
\defset{(\xi,\zeta_0)}{V_0(\xi,\zeta_0) \leq \eps_{0,b}},
\end{equation}
which,
by definition of $\eps_{0,b}$,
is a subset of the basin of attraction of $({\cal P},\KK_0)$.
Then, the supervisor is designed to apply $\KK_0$
as long as $z_0$ is smaller or equal than $\eps_{0,a}$, 
and when is larger or equal to that parameter, a 
jump to $q=1$ is triggered. 
Note that the logic for $q=1$ eventually forces flows for at least $\tau^*$ units of time,
which allows 
$t$  
or 
 $j$ to become large enough, and with that, 
guarantee that $(\xi,\zeta_0)$ is eventually in the set \eqref{eqn:V0set}.
This mechanism is implemented as follows:
\begin{itemize}
\item Flow according to
\begin{equation}
\begin{array}{rcl}
\label{eqn:FlowsAt0a}
\begin{array}{lll}
\dot{\xi}  =  f_p(\xi,\kappa_{c,0}(h_0(\xi),\zeta_0)),  
\dot{\zeta}_0  =  f_{c,0}(h_0(\xi),\zeta_0),\\
\dot{\zeta}_1 = 0,\ \
\dot{z}_0  =  -\eps_0 z_0 + \gamma_0(|h_0(\xi)|), \\
\dot{z}_1  =  0, \ \
\dot{q}  =  0, \ \
\dot{\tau}  =  0
\end{array}
\end{array}
\end{equation}
\begin{equation}
\hspace{-0.2in}\mbox{when }\left\{
\label{eqn:FlowSet0a}\begin{array}{lll}
(\xi,\zeta_0) \in C_0
, \ \ 
\zeta_1 \in \C_1, \ \ \eps_{0,a} \geq z_0 \geq 0,\\ z_1= 0, \ \ q = 0, \ \ 
\tau = 0.\end{array}
\right.
\end{equation}
\item
Jump according to
\begin{equation}
\begin{array}{rcl}
\label{eqn:UpdateLaw0To1a}\begin{array}{lll}
\xi^+ = \xi, \ \
\zeta_0^+  \in \Phi_0, \ \ \zeta_1^+ \in \Phi_1,\ \ z_0^+ = 0,\\
 z_1^+ = 0,\ \ q^+ = 1, \ \ \tau^+ = 0,
\end{array}\end{array}
\end{equation}
when
\begin{equation}
\label{eqn:JumpSet0To1a}
\zeta_1 \in \C_1, \ \ z_0 \geq \eps_{0,a},\ \ z_1 = 0, \ \ q = 0, \ \ \tau=0.
\end{equation}
\end{itemize}

As \eqref{eqn:FlowsAt1a}, the flows defined in \eqref{eqn:FlowsAt0a} enforce, in particular, 
that $q$ remains constant
and that the estimate of $|\xi|_{\A_0}$ converges.
In fact, condition 
\eqref{eqn:FlowSet0a}
allows flows when the estimate of $|\xi|_{\A_0}$
is small enough, permitting it to converge.
When condition \eqref{eqn:JumpSet0To1a} holds, 
a jump back to $q=1$
occurs. As explained below \eqref{eqn:JumpSet1To0a},
in particular, such a jump would occur when, after a jump from $q=1$
to $q = 0$, the state $(\xi,\zeta_0)$ is not in ${\cal B}_0$.
The state $\zeta_1$ is updated to any value
in $\Phi_1$ and the estimator
state $z_1$ is reset to zero.
These selections properly initialize $\KK_1$
and enable our main result in Section~\ref{sec:MainResult}.

\subsubsection{Closed-loop system}
\label{sec:Closed-loopSystem}

We are now ready to write the resulting closed loop as a hybrid system.
The closed-loop hybrid system
has state
$
\chi =
(
{\xi},
{\zeta}_0,
{\zeta}_1,
{z}_0,
{z}_1,
{q},
{\tau})
\in \reals^{n_p} \times \reals^{n_c}\times \reals^{n_c} \times \realsgeq \times \realsgeq \times \Q \times \realsgeq=:X.
$
Collecting the definitions in
Sections~\ref{sec:GlobalController} and \ref{sec:LocalController},
the resulting closed-loop system, which is denoted by
$\HS_{cl}$,
has dynamics given as follows:
$$\begin{array}{rcl}
\dot{\chi}
&=&
\left[
\renewcommand{\arraystretch}{1.5}
\begin{array}{c}
f_p(\xi,\kappa_{c,q}(h_q(\xi),\zeta_q))\\
(1-q)f_{c,0}(h_0(\xi),\zeta_0)\\
q\,f_{c,1}(h_1(\xi),\zeta_1)\\
(1-q)(-\eps_0 z_0 + \gamma_0(|h_0(\xi)|))\\
q(-\eps_1 z_1 + \gamma_1(|h_1(\xi)|))\\
0\\
q
\end{array}
\right]
=: F(\chi), \
\chi \in \widetilde{C},
\end{array}$$
$$\begin{array}{rcl}
\chi^+
& \in &
G_0(\chi)
\cup
G_1(\chi)
\cup
G_s(\chi)
=: G(\chi), \
\chi \in \widetilde{D},
\end{array}$$
where: for each $q = 0$, $(\xi,\zeta_0) \in D_{0}$
$$\begin{array}{rcl}
G_{0}(\chi)  =  \left[
\renewcommand{\arraystretch}{1.5}
\begin{array}{c}
\xi\\
g_{c,0}(h_0(\xi),\zeta_0)\\
\zeta_1\\
(1-\eps_0) z_0 + \gamma_0(|h_0(\xi)|)\\
z_1\\
q\\
\tau
\end{array}
\right],
\end{array}$$
$G_0(\chi)=\emptyset$ otherwise;
for each $q = 1$, $(\xi,\zeta_{1}) \in D_{1}$
$$\begin{array}{rcl}
G_{1}(\chi) =
\left[
\renewcommand{\arraystretch}{1.5}
\begin{array}{c}
\xi\\
\zeta_0\\
g_{c,1}(h_1(\xi),\zeta_1)\\
z_0\\
(1-\eps_1) z_1 + \gamma_1(|h_1(\xi)|)\\
q\\
\tau
\end{array}
\right],
\end{array}$$
$G_1(\chi) = \emptyset$ otherwise;
for each
$\chi \in D_{s,a} \cup D_{s,b}$,
$$\begin{array}{rcl}
G_{s}(\chi)  = 
\left(\xi,\Phi_0,
\Phi_1,
0,
0,
1-q,
0
\right),
\end{array}
$$
$G_s(\chi)=\emptyset$ otherwise;
$$\begin{array}{rcl}
\widetilde{C} &\! \! :=\! \! &  \defset{\chi}{(\xi,\zeta_q) \in C_{q}}
\cap
\left(
C_{s,a} \cup C_{s,b} \cup C_{s,c}
\right),
\\
C_{s,a}&\! \! :=\! \!  &
\defset{\chi}{\zeta_1 \in \C_1, \eps_{0,a} \geq z_0 \geq 0, z_1 = 0, q = 0, \tau=0},\\
C_{s,b}&\! \! := \! \! \! &
\defset{\chi}{\zeta_0 \in \C_0, z_0 = 0, z_1 \geq \eps_{1,a}, q = 1}, \\
C_{s,c} & \! \! := \! \! \! &
\defset{\chi}{\zeta_0 \in \C_0, z_0 = 0, z_1 \geq 0, q = 1, \tau \leq \tau^*},\\
\widetilde{D} & \! \! :=\! \!  &
\defset{\chi}{(\xi,\zeta_q) \in D_{q}}
\cup
D_{s,a}
\cup
D_{s,b},\\
D_{s,a} &\! \!  := \! \! &
\defset{\chi}{\zeta_1 \in \C_1, z_0 \geq \eps_{0,a}, z_1 = 0, q = 0, \tau=0},\\
D_{s,b} & \! \! := \! \! &
\defset{\chi}{\zeta_0 \in \C_0, z_0 = 0,  \eps_{1,a} \geq z_1 \geq 0, q = 1,\tau \geq \tau^*
}.
\end{array}$$
The flow map $F$ is defined in terms of the discrete state $q$ to
``select" the appropriate flow dynamics when $\KK_0$ and $\KK_1$ are applied.
The flow set $\widetilde{C}$ allows flow when both $(\xi,\zeta_q)$ is in the flow set $C_{q}$
and the conditions for flow imposed by the hybrid supervisor
are satisfied.
The latter are given in \eqref{eqn:FlowSet0a}, \eqref{eqn:FlowSet1a}, and \eqref{eqn:FlowSet1b},
which are captured in the sets $C_{s,a}$, $C_{s,b}$, and $C_{s,c}$, respectively.
The jump maps $G_0, G_1,$ and $G_s$ above are defined to execute the jumps of the individual hybrid controllers
when their state jumps
due to 
$(h_q(\xi),\zeta_q) \in D_{c,q}$
or when  reset of the appropriate states is required by the supervisor jump sets
$D_{s,a}$ and $D_{s,b}$,
which are given in \eqref{eqn:JumpSet0To1a} and \eqref{eqn:JumpSet1To0a}, respectively.
Note that since $g_{c,q}$ is only defined on $D_{c,q}$,
the set-valued maps $G_0$ and $G_1$ are nonempty
at points $\chi$ with components in $D_{c,q}$.
For each $i = 0,1$, the functions $\gamma_i$
and constants $\eps_i$
are obtained from the OSS properties of $(\P,\KK_i)$ imposed in 
Assumption~\ref{assumption:1}.
Existence of parameters 
$\eps_{1,a}$ and $\tau^*$
guaranteeing a solution to Problem~($\star$)
 is established in the next section.
 A design method for these parameters is given in
 Section~\ref{sec:DesignProcedure}.

\subsection{Nominal Properties of Closed-loop System}
\label{sec:MainResult}

Our main result is as follows. 

\begin{theorem}{\bf (semi-global asymptotic stability)}
\label{thm:NominalAS}
Suppose Assumption~\ref{assumption:1}  holds.
Then, 
for each compact set 
$\K\subset  X$ of initial conditions
there exists an output-feedback hybrid supervisor $\KK_s$
 such that
the compact set
$$
\A_s
:= \A_0\times \C_0\times \C_1 \times \{0\} \times \{0\} \times \{0\} \times \{0\}$$
is asymptotically
stable for the closed-loop system
$\HS_{cl}$
 with a basin of attraction containing $\K$;
i.e., 
for each $\epsilon>0$ there exists
$\delta>0$ such that each solution $\chi$ to 
$\HS_{cl}$
 with $|\chi(0,0)|_{ \A_s } \leq \delta$ satisfies $|\chi(t,j)|_{\A_s}
 \leq
\epsilon$ for all $(t,j)\in\dom \chi$,
and every maximal solution $\chi$ to $\HS$
with $\chi(0,0)\in \K$ is
complete and satisfies $\lim_{t+j\to\infty} |\chi(t,j)|_{\A_s}=0$.
\end{theorem}
\begin{proof}
By the continuity of $f_p$, $h_i$, and $\kappa_{c,i}$ for each $i = 0,1$ imposed by Assumption~\ref{assumption:1}, and continuity of $\gamma_i$, $F$ is continuous.
By the regularity properties of $g_{c,i}$ guaranteed by well posedness of $\KK_i$ and continuity of $h_i$ from Assumption~\ref{assumption:1}, compactness of $\Phi_i$
for each $i=0,1$, and the definition of 
the set-valued map
$G$,
$G:\widetilde{D}\rightrightarrows X$ is outer semicontinuous, locally bounded, and nonempty for all points in $\widetilde{D}$.
By closedness of $C_q$ and $D_q$ guaranteed by well posedness of $\KK_i$ and continuity of $h_i$,
$\widetilde{C}$ and $\widetilde{D}$ are closed sets. 
This establishes
that the hybrid supervisor is such that the closed-loop system is a well-posed
hybrid system.
Moreover, the construction of $\KK_s$ is such that
solutions to the closed-loop system $\HScl$
exist from points in $\reals^{n_p}\times \reals^{n_c} \times \reals^{n_c} \times \realsgeq\times \realsgeq \times \Q\times \realsgeq$.

Now we show that $\A_0$ is attractive from $\K$.
By the attractivity property induced by $\KK_1$ in Assumption~\ref{assumption:1}.2.a and
Assumption~\ref{assumption:1}.3, for every 
maximal
solution $\chi$ to $\HScl$ from $\widetilde{C} \cup \widetilde{D}$
with $q(0,0) = 1$ there exists $(T,J) \in \dom \chi$ such that $\chi(T,J) \in D_{s,b}$.
By definition of $G$, there exists $J'>J$, $(T,J')\in\dom\chi$ such that $\chi^*=\chi(T,J') \in C_{s,a}$.
Let $\chi'$ be the tail of the 
maximal
solution $\chi$ from $(T',J')$ onwards.
With some abuse of notation, 
every maximal solution $\chi'$ to $\HScl$, with $\chi'(0,0) \in C_{s,a}$ (in particular,
with $\chi'(0,0) = \chi^*$) and
$(\xi'(0,0),\zeta'_0(0,0))$ in the set \eqref{eqn:V0set},
is complete and,
by Assumption~\ref{assumption:1}.1.a and b, satisfies
$
\lim_{t+j \to \infty} |\xi'(t,j)|_{\A_0} = 0.
$
If $(\xi'(t,j),\zeta_0'(t,j))$ never reaches \eqref{eqn:V0set},
we claim that there
exists $(t,j)$ such that
$z_0'(t,j) > \eps_{0,a}$, and then, by the definition of $D$ and $G$, $q$ is mapped to $1$.
Suppose not.  Then, the solution 
$\chi'$
remains in $C_{s,a}$ for all 
$(t',j') \in \dom \chi'$ such that $t'+j' \geq t+j$,
which implies that $z'_0(t',j') \leq \eps_{0,a}$ since 
the norm estimator \eqref{eqn:NormEstimator} for $i=0$ remains on.
Then, 
since $\eps_{0,a} < \eps_{0,b}$,
from \eqref{eqn:VdecreaseEstimator} for $i=0$,
there exists large enough 
$t+j$, $(t,j) \in \dom \chi'$, 
such that $V_0(\xi'(t,j),\zeta'_0(t,j)) \leq \eps_{0,b}$.
This is a contradiction.
Then, a jump to $q=1$ occurs.
By the construction of $C_{s,c}$
and $D_{s,b}$, the closed-loop system will remain 
at $q=1$
for at least $\tau^*$ units of time.
Repeating this argument if needed, the fact that the norm estimator \eqref{eqn:NormEstimator} for $i=1$
guarantees that the estimates converge
(even when reset to zero)
 implies that, eventually, a jump to $q=0$
will occur with $(\xi,\zeta_0)$ in the  set \eqref{eqn:V0set}.
Note that this is the case due to the fact that there are  finitely
many jumps from $q=0$ to $q=1$ and back, as the following result guarantees.

\begin{lemma}\label{claim:finite}
There exist positive parameters $\tau^*$, $\eps_{1,a}$, and $\eps_{1,b}$ 
such that there is no nondecreasing sequence of times\footnote{In the sense that
$t'_{n+1}+j'_{n+1}\geq t'_{n}+j'_{n}$ for all $n \in \nats$.}
$\{(t'_n,j'_n)\}_{n \in \nats} \in \dom \chi$ for which, for all $n \in \nats$, 
\begin{equation}\label{eq:claim:finite}
q(t'_{2n},j'_{2n}) = 0, \qquad q(t'_{2n+1},j'_{2n+1}) = 1.
\end{equation}
\end{lemma}

\begin{proof}
By contradiction, suppose that 
there exist
a complete solution $\chi$ and 
a nondecreasing sequence $\{(t'_n,j'_n)\}_{n \in \nats} \in \dom \chi$,
which can always be chosen so that (\ref{eq:claim:finite}) holds,  
$j'_0 > 0$,
$z_{0}(t'_0,j'_0) = 0$,
flows with $q=0$ occur with $(t,j) \in \dom \chi$, $t\in [t'_{2n},t'_{2n+1})$,
and flows with $q=1$ occur with $(t,j) \in \dom \chi$, $t\in [t'_{2n+1},t'_{2n+2}) $,
$n \in \nats$.
Due to the dynamics of the timer $\tau$, which enforces that 
the time between two jumps from $q=1$ to $q=0$
of the supervisor have $t$'s separated
by at least $\tau^* >0$ seconds, 
the solution cannot be Zeno (see Section~\ref{sec:HybridSystems} for a definition of Zeno solutions). 
It follows that, for each $n \in \nats$,
\begin{equation}\label{eqn:BoundOnZiAtjumps}
\begin{array}{l}
(t'_{n},j'_{n}-1) \in \dom \chi,\  (t'_{n},j'_{n}) \in \dom \chi, 
\\
z_0(t'_{2n+1},j'_{2n +1}-1) \geq \eps_{0,a}, \ 
z_1(t'_{2n+2},j'_{2n+2} -1) \leq \eps_{1,a}.
\end{array}
\end{equation}
By considering the restriction on $\{(t,j) \in \dom \chi'\ : \ \; t \in [t'_{2n+1},t'_{2n+2}], q(t,j) = 1\}$ of the solution $\chi$ as a solution to $(\P,\KK_1)$ issuing from $\chi(t'_{2n+1},j'_{2n+1})$, we get, from \eqref{eqn:VdecreaseEstimator}
with $i = 1$ and \eqref{eqn:BoundOnZiAtjumps}, that
$
V_1(\xi(t'_{2n+2},j'_{2n+2}-1),\zeta_1(t'_{2n+2}, j'_{2n+2} -1)) 
\leq 
z_1(t'_{2n+2},j'_{2n+2}-1)
\nonumber
 + \beta_1(|(\xi(t'_{2n+1},j'_{2n+1}),\zeta_1(t'_{2n+1},j'_{2n+1}))|_{\A_1\times\C_1} $ \\$+ |z_1(t'_{2n+1},j'_{2n+1})|,\delta'_{2n+2})
\leq 
\eps_{1,a}
$ \\$+ \beta_1(|(\xi(t'_{2n+1},j'_{2n+1}),\zeta_1(t'_{2n+1},j'_{2n+1}))|_{\A_1\times\C_1},\delta'_{2n+2}),
$
where
$\delta'_s = (t'_{s} - t'_{s-1}, j'_{s} - j'_{s-1}-1)$
and we have used the fact that $z_1(t'_{2n+1},j'_{2n+1})=0$.
Due to the expression of $G_s$, we have 
$\xi(t'_{2n+2},j'_{2n+2}) = \xi(t'_{2n+2},j'_{2n+2}-1)$
and 
$\zeta_1(t'_{2n+2},j'_{2n+2}) \in \Phi_1$.
Then, 
since $j'_{2n+2}>0$ and $|(\xi,\zeta_1)(t'_{2n+2},j'_{2n+2})|_{\A_1\times \Phi_1}
= |\xi(t'_{2n+2},j'_{2n+2})|_{\A_1}$,
using \eqref{eqn:Vibound}
we obtain
\quad$\\[1em]
\alpha_{1,1}(|(\xi,\zeta_1)(t'_{2n+2},j'_{2n+2})|_{\A_1\times \Phi_1}) \leq$\hfill\null \\[0.7em]\null\hfill$
V_1(\xi(t'_{2n+2},j'_{2n+2} -1),\zeta_1(t'_{2n+2}, j'_{2n+2} -1)).
\quad$\\[1em]
Using \eqref{eqn:BoundOnZiAtjumps}
we get
$\\[1em]
|(\xi,\zeta_1)(t'_{2n+2},j'_{2n+2})|_{\A_1\times \C_1 } \leq$\refstepcounter{equation}\label{eqn:StateBoundZ1}\hfill$(\theequation)$\\[0.5em]\null\hfill$ \alpha_{1,1}^{-1}\left( 
\eps_{1,a} + \beta_1(|(\xi,\zeta_1)(t'_{2n+1},j'_{2n+1})|_{\A_1\times\C_1},\delta'_{2n+2})
\right).\quad$\\[0.7em]
In the same way, 
from \eqref{eqn:VdecreaseEstimator}
with $i = 0$ 
we have, for each $n$,
$
V_0(\xi(t'_{2n+1}, j'_{2n+1} ),\zeta_0(t'_{2n+1},{j'_{2n+1}})) 
\leq 
z_0(t'_{2n+1},j'_{2n+1}-1)
+ \beta_0(|(\xi,\zeta_0)(t'_{2n},j'_{2n})|_{\A_0\times\C_0} + |z_0(t'_{2n},j'_{2n})|,\delta'_{2n+1}).
$
This implies
$\\[1em]
|(\xi,\zeta_0)(t'_{2n+1},j'_{2n+1})|_{\A_0\times \C_0} \leq\alpha_{0,1}^{-1}\left(
z_0(t'_{2n+1},j'_{2n+1}-1)\right.\null\hfill$\\[0.5em]\null\hfill$ \left.
 + \beta_0(|(\xi,\zeta_0)(t'_{2n},j'_{2n})|_{\A_0\times\C_0},\delta'_{2n+1})\right),
\quad$\refstepcounter{equation}\label{eqn:StateBoundZ0}\hfill$(\theequation)$\\[0.7em]
where we have used the fact that $z_0(t'_{2n},j'_{2n})=0$.
Since $q(t'_0,j'_0) = 0$ by construction of the sequence $(t'_n,j'_n)$,
we have, from 
\eqref{eqn:StateBoundZ0},
$
|(\xi,\zeta_0)(t'_{1},j'_{1})|_{\A_0\times \C_0} 
\leq
\alpha_{0,1}^{-1}\left(
z_0(t'_{1},j'_{1}-1) 
+ \beta_0(|(\xi,\zeta_0)(t'_{0},j'_{0})|_{\A_0\times\C_0},\delta'_{1})
\right)
$ and, from
\eqref{eqn:StateBoundZ1},
$
|(\xi,\zeta_1)(t'_{2},j'_{2})|_{\A_1\times \C_1}
\leq
 \alpha_{1,1}^{-1}\left( 
\eps_{1,a} +\right.$ \\$\left. \beta_1(|(\xi,\zeta_1)(t'_1,j'_1)|_{\A_1\times\C_1},\delta'_{2})\right),
$
which implies\\[0.7em]$\null\quad
|(\xi,\zeta_1)(t'_{2},j'_{2})|_{\A_1\times \C_1} \leq$\refstepcounter{equation}\label{ineq:Delta}\hfill$(\theequation)$\\[0.5em]\null\hfill$
 \alpha_{1,1}^{-1}
\left( 
	\eps_{1,a} + \beta_1
 	\left( 
		\Delta+\alpha_{0,1}^{-1}
		\left(
	z_0(t'_{1},j'_{1}-1) \right.\right.\right.$\\[0.5em]\null\hfill$\left.\left.\left.+ \beta_0(|(\xi,\zeta_0)(t'_{0},j'_{0})|_{\A_0\times\C_0},\delta'_{1})\right)
,\delta'_{2}\right)\right),\quad$\\[0.7em]
where 
$\Delta=\max_{x\in\mathcal{A}_0\times \C_0,\; y\in \mathcal{A}_1\times \C_1} |x-y |,$ 
which denotes the maximal distance between the sets $\mathcal{A}_0\times \C_0$ and $\mathcal{A}_1\times \C_1$,
which is finite since both sets are compact.
Consider the compact set $\K$ in the assumption of Theorem \ref{thm:NominalAS}. Due to \eqref{eqn:KLLIOSSboundiOnX}, there exists a compact set containing all solutions of $(\mathcal{P},\KK_0)$ starting from $\K$. By this compactness property, the values  
$
\Delta_1: =\max
\; 
z_0(t,j)$  and $\Delta_2 =\max \; 
|(\xi,\zeta_0)(t,j)|_{\A_0\times\C_0}$
are finite, where the maximum are taken on $\{(t,j)\in \dom \chi :\,  \chi \mbox{{\small{ is a solution of }}}(\mathcal{P},\KK_0)\mbox{{\small{ starting from }}}  \K \}$. Therefore, either there does not exist a sequence
 $\{(t'_n,j'_n)\}_{n \in \nats} \in \dom \chi$ satisfying (\ref{eq:claim:finite}), or there exists such a sequence and inequality (\ref{ineq:Delta}) holds. However, in this latter case, using 
 $\max\{t'_1+j'_1,t'_2+j'_2\} < \tau^*$, 
pick $\tau^*>0$ 
and  $\eps_{1,a}$, $\eps_{1,b}$
 such that
\begin{eqnarray}\non
\alpha_{1,2}(\alpha_{1,1}^{-1}
\left( 
	\eps_{1,a} + \overline{\beta}_1
 	\left( 
		\Delta+\alpha_{0,1}^{-1}
		\left(
			\Delta_1
						+ \overline{\beta}_0(\Delta_2,\tau^*)\right)
,\tau^*\right)\right) & & \\ 
& & \hspace{-0.7in} \leq \eps_{1,b}\label{cp1}
\end{eqnarray}
where 
$\overline{\beta}_i(s,t) := 2\exp(-\eps_i t) \alpha_{i,2}(s)$, $i = 0,1$.
Then, using (\ref{eqn:Vibound}) with $i=1$, we get 
$V_1(\xi(t'_{2},j'_{2}),\zeta_1(t'_{2},j'_{2}))\leq \eps_{1,b}$. With (\ref{eqn:V1set}), 
since $\zeta_1(t'_2,j'_2) \in \Phi_1$ and $\zeta_0(t'_2,j'_2) \in \Phi_0$,
we get $\gamma_0(h_0(\xi(t'_{2},j'_{2}))) <\eps_{0,a}\eps_0$. 
Since the supervisor uses $\KK_0$ at $(t'_2,j'_2)$
and Assumption~\ref{assumption:1}.3 implies that, along solutions,
$\dot{z}_0 < -\eps_0 z_0 + \eps_{0,a}\eps_0$, 
we have that $z_0(t,j) < \eps_{0,a}$ for all future $(t,j)$.
Then, no future jump of the supervisor is possible, which is
a contradiction, showing that there is no sequence satisfying \eqref{eq:claim:finite}.
\end{proof}

By the attractivity properties of the basin of attraction 
of
and completeness of solutions to 
$(\P,\KK_0)$, it follows
that every 
maximal
solution converges to $\A_s$.
Hence, solutions are bounded.
By the construction of the jump map in equation \eqref{eqn:UpdateLaw1To0a},
the state $\zeta_1$ converges to $\Phi_1$ while
$z_1$ and $\tau_0$ converge to zero.
To conclude the proof, note that
the local stability properties induced by $\KK_0$ 
and the property $h_0(\A_0) = 0$ imply  
that $\A_s$ is stable.
\end{proof}

\begin{remark}
Note that 
when assuming the existence of a norm-observer for $\mathcal{P}$ (and not a pair of norm-observers for $\mathcal{P}$ in closed loop with $\mathcal{K}_0$ and with $\mathcal{P}$ in closed loop with $\mathcal{K}_1$ 
as in Assumption~\ref{assumption:1}), we  
obtain
 a globally asymptotic stabilizing hybrid controller $\KK_s$. Indeed, following the proof of Theorem \ref{thm:NominalAS} with this additional assumption, we may strength the result of Lemma \ref{claim:finite} and obtain that there is no nondecreasing sequence of times satisfying (\ref{eq:claim:finite}) {\it for any initial condition} (globally).
With such a detectability assumption, the obtained result would be close in spirit to \cite{PrieurTeel:ieee:11}, but generalizes 
it
since \cite{PrieurTeel:ieee:11} pertains to
the problem of uniting
continuous-time controllers 
with same objectives.
\end{remark}

%%%%%%%%%%%%%%%%%%%%%%%%%%%%%%%%%%%%%%
%%%%% DESIGN PROCEDURE
%%%%%%%%%%%%%%%%%%%%%%%%%%%%%%%%%%%%%%%

\subsection{A Design Procedure}
\label{sec:DesignProcedure}

Theorem~\ref{thm:NominalAS}
guarantees the existence of an output-feedback hybrid supervisor
solving Problem ($\star$).
While this result does not explicitly provide 
values of the supervisor parameters,
the steps in its proof provide guidelines
(potentially conservative)
on how to choose these parameters.  
When exponential-decay OSS-Lyapunov functions
and associated functions
certifying the OSS properties in Assumption~\ref{assumption:1} 
are available (see \eqref{eqn:Vibound}-\eqref{eqn:ViChangeBound}),
the design procedure in the following result
is a consequence of the arguments in the proof
of Theorem~\ref{thm:NominalAS}.
\begin{corollary}{\bf (design procedure)}
\label{coro:design}
Suppose Assumption~\ref{assumption:1} holds.
The output-feedback hybrid supervisor $\KK_s$
with parameters $\eps_{0,a}$, $\eps_{1,b}$, and $\tau^*$
designed following the next steps
solves Problem ($\star$).
\begin{list}{}{}
\item[Step 1)]
Let $\eps_{0,b}>0$ such that 
$\Gamma_0\!:=\! \defset{(\xi,\zeta_0)\!\!}{\!\!V_0(\xi,\zeta_0) \leq \eps_{0,b}}$
 is a subset
of the basin of attraction ${\cal B}_0$
for the asymptotic stabilization of 
$\A_0\times\C_0$ 
with $\KK_0$.
\item[Step 2)]
Choose $\eps_{0,a}>0$ and $\eps_{1,b}>0$ 
so that $\eps_{0,a} < \eps_{0,b}$,
$
\Gamma_1:= \defset{\xi \in \reals^{n_p}}{V_1(\xi,\zeta_1) \leq \eps_{1,b},\ \zeta_1 \in \C_1} \times
\C_0
$
is a subset of $\Gamma_0$,
and every solution 
$(\xi,\zeta_0)$ to $(\P,\KK_0)$
from $\Gamma_1$
satisfies
$\gamma_0(|h_0(\xi(t,j))|) < \eps_{0,a}\eps_0$ 
for all $(t,j)\in\dom (\xi,\zeta_0)$.
\item[Step 3)] Design $\eps_{1,a}>0$ and $\tau^*>0$ 
such that
\begin{eqnarray*}\non
\hspace{-0.3in}\alpha_{1,2}(\alpha_{1,1}^{-1}
\left( 
	\eps_{1,a} + \overline{\beta}_1\!\left( 
		\Delta+\alpha_{0,1}^{-1}\left(
			\Delta_1 + \overline{\beta}_0(\Delta_2,\tau^*)\right),\tau^*\right)\right) & & \\ 
& & \hspace{-0.2in} \leq \eps_{1,b}.
\end{eqnarray*}
where
$\Delta=\max_{x\in\mathcal{A}_0\times \C_0,\; y\in \mathcal{A}_1\times \C_1} |x-y |$,
$\Delta_1 =\max
\; 
z_0(t,j)$, $\Delta_2 =\max \; 
|(\xi,\zeta_0)(t,j)|_{\A_0\times\C_0}$ 
for each  solution $(\xi,\zeta_0)$ to $(\P,\KK_0)$
from $\M$ (projected onto $\reals^{n_p} \times \reals^{n_c}$),
and
$\overline{\beta}_i(s,t) = 2\exp(-\eps_i t) \alpha_{i,2}(s)$
for each $i = 0,1$.
\end{list}
\end{corollary}
Note that the condition in Step 3
can always be satisfied 
by picking small enough parameter $\eps_{1,a}$,
which defines the threshold for $z_1$ to switch from $q = 1$ to $q=0$,
and large enough parameter $\tau^*$, which forces flows with controller $\KK_1$
until the timer reaches such value.
Such selections have the effect of enlarging the time
the controller $\KK_1$ is in the loop, making it possible that, after a jump from
$q=1$ to $q=0$,
the state of the plant is such that controller $\KK_0$
stabilizes $\A_0\times\C_0$ without further jump back to $q=1$.
Note that the condition in Step 3 is a consequence of the 
proof of Lemma~\ref{claim:finite}, which guarantees that 
there are finitely many jumps from $q = 0$ to $q = 1$ and back
(but does not quantify the number of such jumps).
The design procedure and, in particular, the tuning of $\eps_{1,a}$ and $\tau^*$
are illustrated in Section~\ref{ex:1} when
revisiting Example~\ref{ex:LimitedInformation1}.

%%%%%%%%%%%%%%%%%%%%%%%%%%%%%%%%%%%%%%%%

\subsection{Robustness of the Closed-loop System}

The following model of the plant with perturbations is considered
\begin{equation}\label{eqn:PPerturbed}
\dot{\xi} = f_p(\xi,u+d_1) + d_2
\end{equation}
with outputs $y_{p,0} = h_0(\xi) + d_3$
and $y_{p,1} = h_1(\xi) + d_4$, where 
$d_1$ corresponds to actuator error, 
$d_2$ captures unmodeled dynamics,
and $d_3,d_4$ represent measurement noise.\footnote{
The exogenous signals $d_i$, $i=1,\ldots,4$, 
are given
on hybrid time domains, and in general, their
value can jump at jump times. 
For exogenous signals $d_i(t)$, that is, given by functions of time,
given a hybrid time domain $S$ it is possible to define, with some abuse of
notation, $d_i(t,j):=d_i(t)$ for each $(t,j) \in S$. Solutions to
hybrid systems with the perturbations above is understood similarly
to the concept of solution defined in Section~\ref{sec:HybridSystems}.}
Then,
denoting by $\widetilde{d}_i$
the signals $d_i$ extended to the state space of $\chi$,
the overall closed-loop system $\HS_{cl}$
results in a perturbed hybrid system, which is denoted by $\widetilde{\HS}_{cl}$, 
with dynamics
$$\begin{array}{rcl}
\begin{array}{clll}
\dot{\chi} &=& F(\chi+\widetilde{d}_1) + \widetilde{d}_2 &\ \ \chi + \widetilde{d}_1 \in \widetilde{C} \\
{\chi}^+ &\in& G(\chi+\widetilde{d}_1) + \widetilde{d}_2 &\ \ \chi +\widetilde{d}_1 \in \widetilde{D}\ .
\end{array}
\end{array}
$$
The following 
result asserts that the 
stability
of the closed-loop system 
is robust to a class of perturbations.
It follows from the asymptotic stability property established
in Theorem~\ref{thm:NominalAS}
and the fact that the construction of the hybrid supervisor
leads to a well-posed closed-loop system.

\begin{theorem}{\bf (stability under perturbations)}
\label{thm:robustness}
Suppose Assumption~\ref{assumption:1} holds.
Then, 
there exists $\beta \in \classKLL$
such that,
for each $\eps >0$ and each compact set
$\K\subset  X$,
there exists $\delta>0$ such that for each measurable
$\widetilde{d}_1,\widetilde{d}_2:\realsgeq \to \delta\ball$
every solution
$\chi$ to $\widetilde{\HS}_{cl}$ 
with $\chi(0,0) \in \K$
satisfies\\
$
\phantom{GA}|\chi(t,j)|_{\A_s}\leq \beta(|\chi(0,0)|_{\A_s}, t,j)+\eps\  \ \ \forall (t,j) \in \dom \chi.
$
\end{theorem}
\begin{proof}
 By Theorem 6.5 in \cite{GoebelTeel06}, there exists
 $\beta \in \classKLL$ such that all solutions $\chi$ to $\HScl$
  satisfy
$
|\chi(t,j)|_{\A_s}\leq \beta(|\chi(0,0)|_{\A_s},t,j)
$
for all $(t,j)\in \dom \chi$.
Consider the perturbed hybrid system $\widetilde{\HS}_{cl}$.
Since $\widetilde{d}_1(t),\widetilde{d}_2(t) \in \delta\ball$ for all $t\geq 0$,
the closed-loop system
$\widetilde{\HS}_{cl}$  can be written as
\begin{equation}\label{eqn:HSclPerturbed}
\begin{array}{rcl}
\mattarraytwo{
\dot{\chi} &\in F_{\delta}(\chi)\ \ \ \ \chi \in C_{\delta}\\
\chi^+ &\in G_{\delta}(\chi)\ \ \ \ \chi \in D_{\delta},}
\end{array}
\end{equation}
where
$
F_{\delta}(\chi)  :=  \cco F(\chi+\delta\ball) + \delta \ball$,\\ $\non
G_{\delta}(\chi) :=  \defset{\eta}{\eta \in \chi' + \delta\ball, \chi' \in G(\chi+\delta\ball)}$, \\$\non
C_{\delta}   :=  \defset{\chi}{(\chi+\delta\ball) \cap \widetilde{C} \not= \emptyset}$, 
 and \\$\non
D_{\delta}  :=  \defset{\chi}{(\chi+\delta\ball) \cap \widetilde{D} \not= \emptyset}$.
This hybrid system corresponds to
an outer perturbation of $\HScl$ and satisfies (C1), (C2), (C3), and
(C4) in \cite{GoebelTeel06} (see Example 5.3 in \cite{GoebelTeel06} for more
details).
Then, the claim follows by Theorem 6.6 in \cite{GoebelTeel06} since, for each compact
set $\K$ of the state space and each $\eps>0$, there exists $\delta^*>0$ such
that for each $\delta \in (0,\delta^*]$, every solution $\chi_{\delta}$
to 
\eqref{eqn:HSclPerturbed}
 from $\K$ satisfy, for all $(t,j)\in\dom \chi_{\delta}$,
$
|\chi_{\delta}(t,j)|_{\A_s} \leq \beta(|\chi_{\delta}(0,0)|_{\A_s},t,j)+\eps.
$
\end{proof}

\begin{remark}
\label{rmk:Extensions}
The stability and attractivity assumptions imposed in
Theorem~\ref{thm:NominalAS} and Theorem~\ref{thm:robustness}
can be further relaxed as in \cite{PrieurTeel:ieee:11}.
In particular, the attractivity induced
by $\KK_1$ can be relaxed to be semi-global and practical (by adapting the considered compact set $\K\subset  X$ to these ``semi-global and practical'' properties).
Also, it can be relaxed to allow the individual controllers to have solutions that are bounded 
but not complete, as long as the solutions to the closed-loop system are all complete.
Lastly, note that Theorem~\ref{thm:robustness} gives a qualitative robustness result. When focusing on specific nonlinear systems (such as linear systems with saturation at the input), estimations of basins of attraction of individual continuous-time controllers have been used in \cite{PrieurTeel:ieee:11} and thus it may be possible, for this class of specific nonlinear systems, to derive qualitative results and more explicit bounds for the robustness issue. 
\end{remark}

\section{Examples}
\label{sec:examples}

The proposed control algorithm piecing together two output-feedback
hybrid controllers is applicable to numerous control systems where 
the design of a single robust  stabilizing controller 
is difficult or even impossible.  Such applications include the stabilization of the inverted
position of the single pendulum \cite{SanfeliceTeelGoebelPrieur06ACC},
the inverted position of the pendubot \cite{SanfeliceTeel07ACC},
the position and orientation of a mobile robot \cite{Sanfelice.ea.08.CDC.Supervisor},
and the synchronization of Lorenz oscillators \cite{Efimov.ea.IJRNC.11}.
An implementation of the proposed controller in a real-world system
will result in a logic-based algorithm that triggers the 
discrete updates of the variables $z_0$, $z_1$, $q$, and $\tau$
by checking via {\em if/else} statements if the variables and measurements 
are in the jump set $\widetilde{D}$.  
In such situations, the algorithm will update the 
values of the variables at the next time step.
For an example of such an implementation, see \cite{OFlahertySanfeliceTeel07}.

Next, we revisit Examples \ref{ex:Avoidance1} and \ref{ex:LimitedInformation1}.

\subsection{Stabilization with constrained inputs and limited information}
\label{ex:1}

Consider the  stabilization of the origin of \eqref{eqn:plantEx1} in Example~\ref{ex:LimitedInformation1}.
Suppose that the inputs are constrained to $u_1u_2=0$
and that $\overline{\alpha}$ is a constant satisfying $|\overline{\alpha}| \in (0,\frac{\sqrt{3}}{3})$.
Measurements of $\xi_1$ and $\xi_2$ are available but not simultaneously.
Due to these constraints, the task of 
designing a single controller or a controller uniting
two controllers with the same objectives 
for the stabilization of the origin is daunting.
However, 
a hybrid controller $\KK_s$, as presented in this paper, can be designed
to accomplish this task by coordinating two controllers, $\KK_0$ and $\KK_1$,
with different objectives.
Consider the  controller $\KK_0$ 
in Example~\ref{ex:LimitedInformation1}
which consists of a
static feedback controller
that measures $h_0(\xi) := \xi_1$ to stabilize $\xi$ to $\A_0 = (0,0)$.
From \eqref{eqn:V0dot},
it can be verified that $\defset{\xi}{V_0(\xi)\leq \frac{1}{6}} \subset {\cal B}_0$,
with ${\cal B}_0$ being the basin of attraction 
for $\KK_0$.  
Since $|\overline{\alpha}| \in (0,\frac{\sqrt{3}}{3})$, we have that $V_0((0,\overline{\alpha})) < \frac{1}{6}$ and thus the point $(0,\overline{\alpha})$ is in the interior of ${\cal B}_0$.
A  controller $\KK_1$ can be designed to steer the solutions to $\A_1:=(0,\overline{\alpha})$.
From \eqref{eqn:V0dot}, 
it follows that the point $(0,\overline{\alpha})$ 
belongs to the interior of ${\cal B}_0$; 
hence item 3 in Assumption~\ref{assumption:1} holds.
Let $h_1(\xi) := \xi_2-\overline{\alpha}$.
The controller $\KK_1$
is given as in \eqref{eq:controller} with $n_c=0$, $\kappa_{c,1}(\xi):= [h_1(\xi)+\overline{\alpha},\ 0]^{\top}$, 
and no dynamical state (i.e., $C_{c,1} = D_{c,1} = \emptyset$ and $f_{c,1}, g_{c,1}$ are
arbitrary).
With this controller, 
the function $V_1(\xi) = \frac{1}{4}\xi_1^4+\frac{1}{2}(\xi_2-\overline{\alpha})^2$ satisfies,
for all $\xi \in \reals^2$,
$
\langle \nabla V_1(\xi),f_p(\xi,\kappa_{c,1}(\xi)) \rangle
\leq -V_1(\xi),$
from where a norm observer for $|\xi|_{\A_1}$ follows;
e.g., we can use $\dot{z}_1 = - z_1$.
Then, Assumption~\ref{assumption:1} holds
with 
$m_{c,0} = m_{c,1} = 1$,
$\C_0 = \C_1 = \emptyset$,
$\eps_0 = 1$,
and
$\eps_1 = 1$.
Then, using Theorem~\ref{thm:NominalAS}
there exists a hybrid supervisor $\KK_s$ such that
the origin of \eqref{eqn:plantEx1} is asymptotically stable.
Following Section~\ref{sec:Closed-loopSystem},
the closed-loop system
has state
$
\chi =
(
{\xi},
{z}_0,
{z}_1,
{q},
{\tau})
\in \reals^{2} 
\times \reals \times \reals \times \Q \times \reals=:X
$
and is given by\footnote{
We denote the $i$-th component
of $\kappa_{c,q}$ by 
$\kappa_{c,0}^i(\xi)$, $i = 1,2$, $q = 0,1$.
}
$$
 F(\chi) :=
\left[
\renewcommand{\arraystretch}{1.5}
\begin{array}{l}
\left[
\begin{array}{lll}
-\xi_1+(\kappa_{c,q}^1(\xi)- \xi_2)\xi^2_1\\
-\xi_2+\xi^2 _1+ \overline{\alpha} + \kappa_{c,q}^2(\xi)
\end{array}
\right]\\
(1-q)(-z_0 + |h_0(\xi)|^4 (1+|h_0(\xi)|^2))\\
- q\, z_1\\
0\\
q
\end{array}
\right], 
$$
$
G(\chi)  :=
[
\xi^\top\ \
0\ \
0\ \
1-q\ \ 
0
]^\top,
$
$\widetilde{C} :=  
C_{s,a} \cup C_{s,b} \cup C_{s,c}$,
$$\begin{array}{rcl}
C_{s,a}& :=  &
\defset{\chi}{
\eps_{0,a} \geq z_0 \geq 0, z_1 = 0, q = 0, \tau=0},\\
C_{s,b}& :=  &
\defset{\chi}{
z_0 = 0, z_1 \geq \eps_{1,a}, q = 1}, \\
C_{s,c} & :=  &
\defset{\chi}{
z_0 = 0, z_1 \geq 0, q = 1, \tau \leq \tau^*},
\end{array}$$
$\widetilde{D}  := 
D_{s,a}
\cup
D_{s,b}$,
$$\begin{array}{rcl}
D_{s,a} & := &
\defset{\chi}{
z_0 \geq \eps_{0,a}, z_1 = 0, q = 0, \tau=0},\\
D_{s,b} & := &
\defset{\chi}{
z_0 = 0,  \eps_{1,a} \geq z_1 \geq 0, q = 1, \tau \geq \tau^*}.
\end{array}$$
Figure~\ref{fig:1} shows a trajectory to the closed-loop system
when $\overline{\alpha} = \frac{1}{4}$, $\eps_{0,a} = \eps_{1,a} = 0.01$, $\tau^* = 1$,
and $\K_0 = 10\ball$,
which are parameters found numerically.
The trajectory starts from 
$\xi(0,0) = (3,-3)$
with controller $\KK_1$ 
connected to 
the plant ($q = 1$), which steers the plant component to 
a neighborhood of the origin.
At about $(t,j) \approx (4.65,0)$,
$z_1$ reaches $\eps_{1,a}$ and $\tau$ is above $\tau^*$,
triggering a jump to $q =0$. 
In that mode,
the local controller steers the plant component to zero,
$z_0$ approaches zero, and the other controller components
remain at zero.
Figure~\ref{fig:1bis}
shows a trajectory to the closed-loop system
with $q(0,0) = 0$
and $\xi(0,0) = (30,-30)$.
In this case, a jump of the supervisor to $q=1$ occurs initially.\footnote{Dashed (red) lines denote jumps in the state components.}
Since after the jump $z_1$ is mapped to zero,
$z_1$ remains at zero for the remainder of the solution,
jumps back to $q=0$ are triggered every $\tau^*$ seconds,
with instantaneous jumps back to $q=1$ until the local
controller is capable of stabilizing $\A_0$.

The design procedure in Corollary~\ref{coro:design}
can be used to systematically select
parameters $\eps_{1,a}$ and $\tau^*$.
In this way, we follow the steps proposed therein
with $\bar{\alpha} = \frac{1}{4}$ and
$\K_0 = 10\ball$.
Since, as shown earlier, we have $\defset{\xi}{V_0(\xi)\leq \frac{1}{6}} \subset {\cal B}_0$,
then we pick $\eps_{0,b} = \frac{1}{6}$ in Step 1 and define $\Gamma_0$.
When $\frac{4}{27} \leq \eps_{0,a} <  \eps_{0,b}$ and $\eps_{1,b} \leq 0.015$,
we have that the conditions in Step 2 hold.
In fact, solutions $\xi$ from $\Gamma_0$ satisfy $|\xi(t,j)|\leq \frac{\sqrt{3}}{3}$ for all $(t,j) \in \dom \xi$ and, since $\gamma_0(s) = s^4(1+s^2)$, we have $\gamma_0(|h_0(\xi(t,j))|) \leq \frac{4}{27}$.
Moreover, a simple check on level sets indicates that
$\Gamma_1:=\defset{\xi \in \reals^{2}}{V_1(\xi) \leq 0.015} \subset \Gamma_0$.
To pick $\eps_{1,a}$ and $\tau^*$ in Step 3, we first obtain the 
following values after straightforward computations:
$\Delta = |\bar{\alpha}|$, $\Delta_1 = \eps_{0,a}$,
$\Delta_2 = \alpha_{0,1}^{-1}\left(\eps_{0,a}+
2 \alpha_{0,2}(10+\eps_{0,a})\right)$,
$\alpha_{0,1}^{-1}(s) = (2s)^{1/2}$,
$\alpha_{0,2}(s) = \frac{1}{2}s^2$, and
$\alpha_{1,1}^{-1}(s) = 2 \max\left\{s^{1/4},s^{1/2} \right\}$.
Using $\eps_{0,a} = \frac{4}{27}$,
then the condition in Step 3 is satisfied 
with $\eps_{1,a}= 0.00005$ and $\tau^* = 15$.
Figure~\ref{fig:1bisbis} shows a simulation
of the closed-loop system with these parameters,
which indicates that convergence to the origin
occurs after only one jump.

%\begin{figure}[h!]  
%\begin{center}  
%\psfrag{x}[][][0.9]{\hspace{-0.25in} $\xi_1$}
%\psfrag{y}[][][0.9]{\qquad$\xi_2$}
%\psfrag{q}[][][0.9]{$q$}
%\psfrag{tau}[][][0.9]{$\tau$}
%\psfrag{z1}[][][0.9]{$z_1$}
%\psfrag{z0}[][][0.9]{$z_0$}
%\psfrag{flows [t]}[][][0.7]{$t$}
%\subfigure[Plant trajectory \label{fig:PlanarPlot-LimitedInfo}]
%{\includegraphics[width=.47\textwidth]{plane}}  
%\subfigure[Controller trajectory \label{fig:EstimatorPlot-LimitedInfo}]
%{\includegraphics[width=.52\textwidth]{ControllerStates}}  
%\end{center}  
%\caption{Plant and controller states of a closed-loop trajectory.  (a) Plant component $\xi(t,j)$ for \eqref{eqn:plantEx1} from $\xi(0,0)=(3,-3)$,
%$q(0,0) = 1$, $\tau(0,0) = z_0(0,0)= 0$,
%$z_1(0,0) = 1$.  
%Dotted lines denote an estimate of ${\cal B}_0$, 
%$\star$ (red) the jump from $q=1$ to $0$,
%and $\times$ the sets $\A_1 = (0,\overline{\alpha})$ and $\A_0 = (0,0)$,
%with $\overline{\alpha} = \frac{1}{4}$.
%(b) Controller states of hybrid supervisor $\KK_s$.
%The dashed lines represent the jumps in the variables.
%Controller parameters: 
%$\eps_{0,a} = \eps_{1,a} = 0.01$, and $\tau^* = 1$.}
%\label{fig:1}
%\end{figure}
%-----------------------
\begin{figure}[h!]  
\begin{center}  
\subfigure[Plant trajectory \label{fig:PlanarPlot-LimitedInfo}]
{
\psfragfig*[width=.47\textwidth]{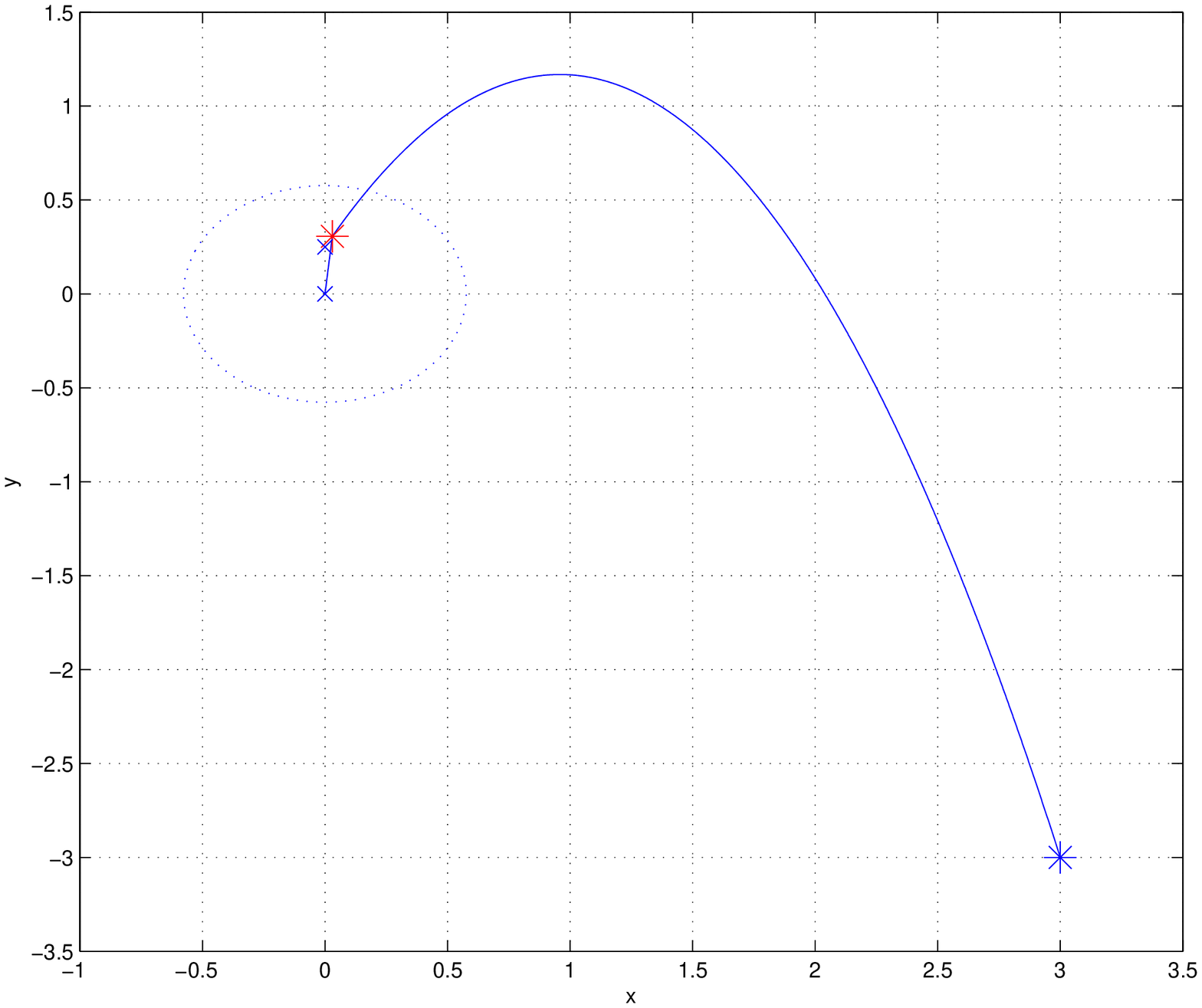}
{
\psfrag{x}[][][0.9]{\hspace{-0.25in} $\xi_1$}
\psfrag{y}[][][0.9]{\qquad$\xi_2$}
}
}
\subfigure[Controller trajectory \label{fig:EstimatorPlot-LimitedInfo}]
{
\psfragfig*[width=.47\textwidth]{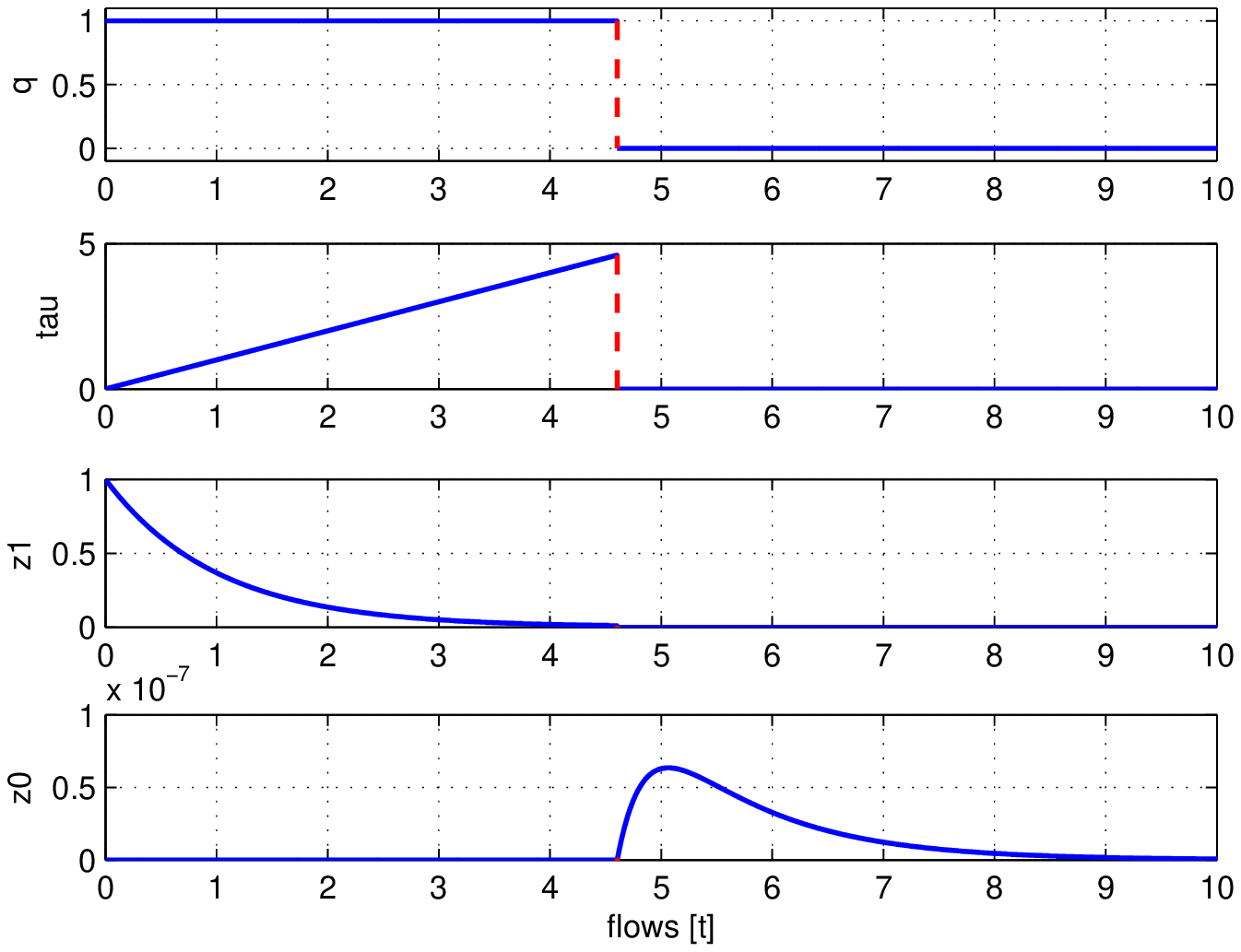}
{
\psfrag{x}[][][0.9]{\hspace{-0.25in} $\xi_1$}
\psfrag{y}[][][0.9]{\qquad$\xi_2$}
\psfrag{q}[][][0.9]{$q$}
\psfrag{tau}[][][0.9]{$\tau$}
\psfrag{z1}[][][0.9]{$z_1$}
\psfrag{z0}[][][0.9]{$z_0$}
\psfrag{flows [t]}[][][0.7]{$t$}
}  
}
\end{center}  
\caption{Plant and controller states of a closed-loop trajectory.  (a) Plant component $\xi(t,j)$ for \eqref{eqn:plantEx1} from $\xi(0,0)=(3,-3)$,
$q(0,0) = 1$, $\tau(0,0) = z_0(0,0)= 0$,
$z_1(0,0) = 1$.  
Dotted lines denote an estimate of ${\cal B}_0$, 
$\star$ (red) the jump from $q=1$ to $0$,
and $\times$ the sets $\A_1 = (0,\overline{\alpha})$ and $\A_0 = (0,0)$,
with $\overline{\alpha} = \frac{1}{4}$.
(b) Controller states of hybrid supervisor $\KK_s$.
The dashed lines represent the jumps in the variables.
Controller parameters: 
$\eps_{0,a} = \eps_{1,a} = 0.01$, and $\tau^* = 1$.}
\label{fig:1}
\end{figure}

%-----------------------
%\begin{figure}[h!]  
%\begin{center}  
%\psfrag{x}[][][0.9]{\hspace{-0.25in} $\xi_1$}
%\psfrag{y}[][][0.9]{\qquad$\xi_2$}
%\psfrag{q}[][][0.9]{$q$}
%\psfrag{tau}[][][0.9]{$\tau$}
%\psfrag{z1}[][][0.9]{$z_1$}
%\psfrag{z0}[][][0.9]{$z_0$}
%\psfrag{flows [t]}[][][0.7]{$t$}
%\psfrag{0.005}[][][0.9]{}
%\subfigure[Plant trajectory and zoom around origin\label{fig:2PlanarPlot-LimitedInfo}]
%{\includegraphics[width=.475\textwidth]{plane2complete}}  
%\subfigure[Controller trajectory \label{fig2:EstimatorPlot-LimitedInfo}]
%{\includegraphics[width=.51\textwidth]{ControllerStates2}}  
%\end{center}  
%\caption{Plant and controller states of a closed-loop trajectory.  (a) Plant component $\xi(t,j)$ for \eqref{eqn:plantEx1} from $\xi(0,0)=(30,-30)$,
%$q(0,0) = \tau(0,0) = z_0(0,0)= 0$, 
%$z_1(0,0) = 1$.  
%Dotted lines denote an estimate of ${\cal B}_0$, 
%$\star$ (red) the jump from $q=1$ to $0$,
%and $\times$ the sets $\A_1 = (0,\overline{\alpha})$ and $\A_0 = (0,0) (= \A)$,
%with $\overline{\alpha} = \frac{1}{4}$.
%(b) Controller states of hybrid supervisor $\KK_s$.
%The dashed lines represent the jumps in the variables.
%Controller parameters: 
%$\eps_{0,a} = \eps_{1,a} = 0.01$, and $\tau^* = 1$.
%}
%\label{fig:1bis}
%\end{figure}
%-----------------------
\begin{figure}[h!]  
\begin{center}  
\subfigure[Plant trajectory and zoom around origin\label{fig:2PlanarPlot-LimitedInfo}]
{
\psfragfig*[width=.455\textwidth]{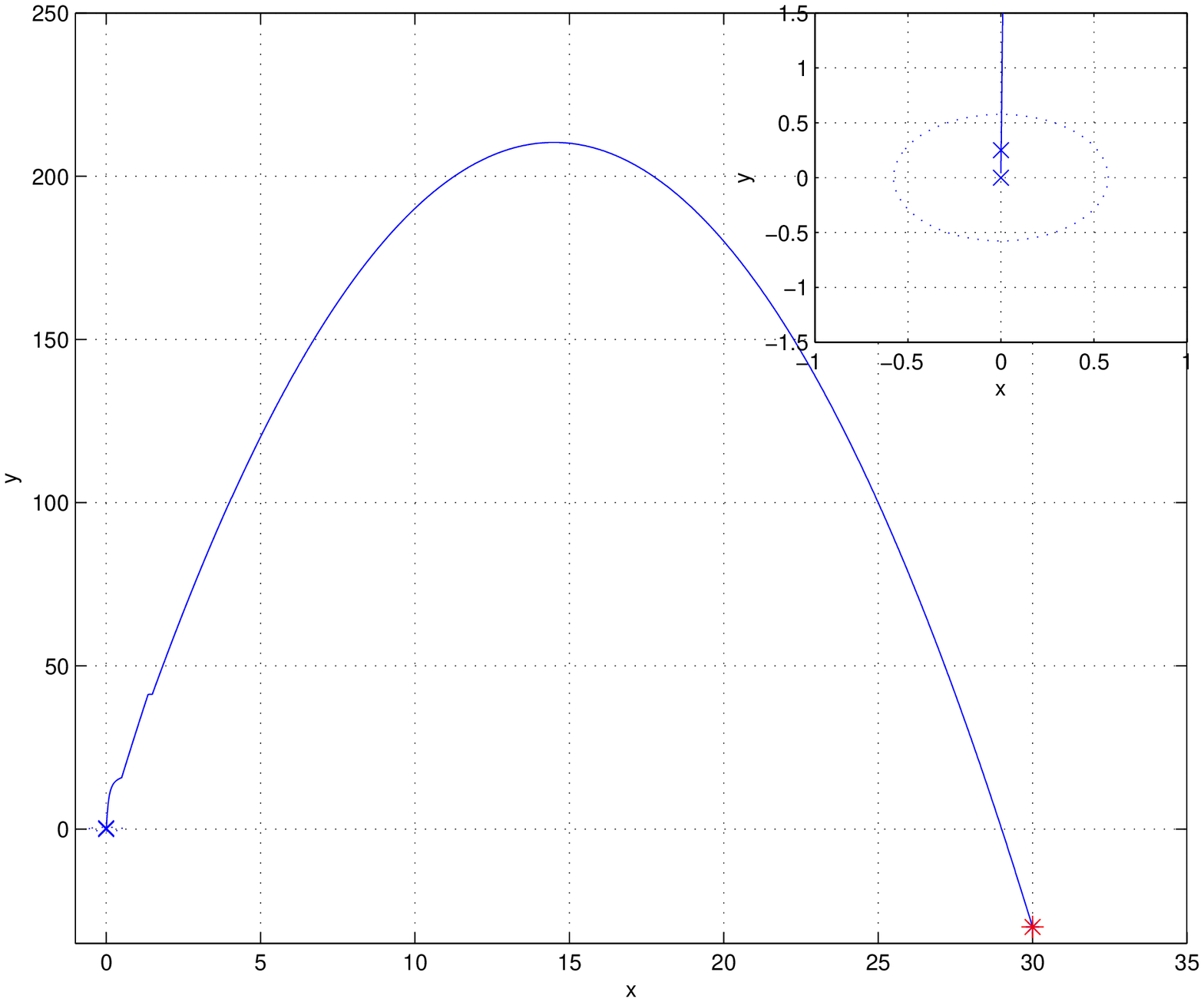}
{
\psfrag{x}[][][0.9]{\hspace{-0.25in} $\xi_1$}
\psfrag{y}[][][0.9]{\qquad$\xi_2$}
\psfrag{q}[][][0.9]{$q$}
\psfrag{tau}[][][0.9]{$\tau$}
\psfrag{z1}[][][0.9]{$z_1$}
\psfrag{z0}[][][0.9]{$z_0$}
\psfrag{flows [t]}[][][0.7]{$t$}
\psfrag{0.005}[][][0.9]{}
}
}  
\subfigure[Controller trajectory \label{fig2:EstimatorPlot-LimitedInfo}]
{
\psfragfig*[width=.45\textwidth]{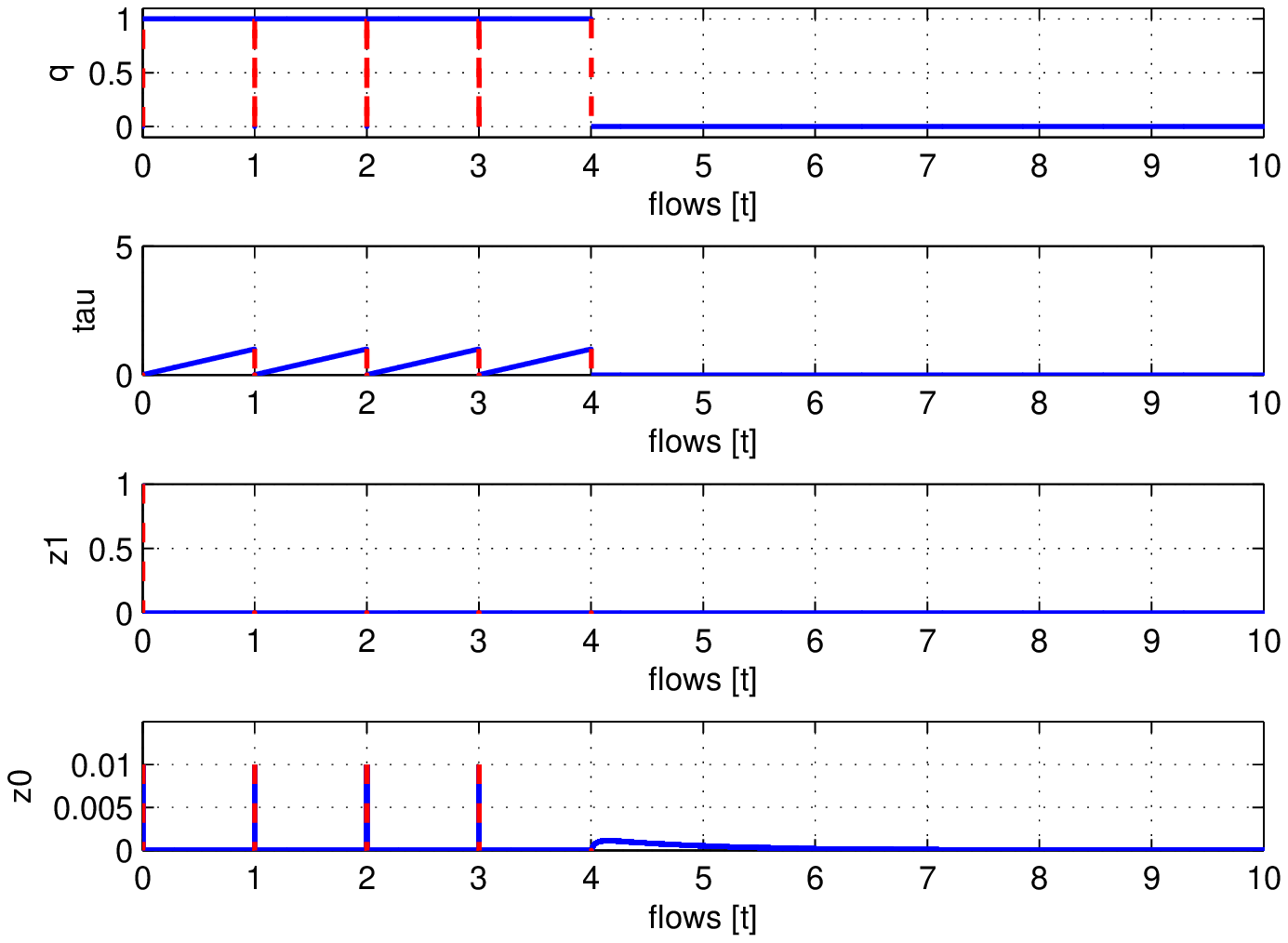}
{
\psfrag{x}[][][0.9]{\hspace{-0.25in} $\xi_1$}
\psfrag{y}[][][0.9]{\qquad$\xi_2$}
\psfrag{q}[][][0.9]{$q$}
\psfrag{tau}[][][0.9]{$\tau$}
\psfrag{z1}[][][0.9]{$z_1$}
\psfrag{z0}[][][0.9]{$z_0$}
\psfrag{flows [t]}[][][0.7]{$t$}
\psfrag{0.005}[][][0.9]{}
}
}  
\end{center}  
\caption{Plant and controller states of a closed-loop trajectory.  (a) Plant component $\xi(t,j)$ for \eqref{eqn:plantEx1} from $\xi(0,0)=(30,-30)$,
$q(0,0) = \tau(0,0) = z_0(0,0)= 0$, 
$z_1(0,0) = 1$.  
Dotted lines denote an estimate of ${\cal B}_0$, 
$\star$ (red) the jump from $q=1$ to $0$,
and $\times$ the sets $\A_1 = (0,\overline{\alpha})$ and $\A_0 = (0,0) (= \A)$,
with $\overline{\alpha} = \frac{1}{4}$.
(b) Controller states of hybrid supervisor $\KK_s$.
The dashed lines represent the jumps in the variables.
Controller parameters: 
$\eps_{0,a} = \eps_{1,a} = 0.01$, and $\tau^* = 1$.
}
\label{fig:1bis}
\end{figure}

%-----------------------

%\begin{figure}[h!]  
%\begin{center}  
%\psfrag{x}[][][0.9]{\hspace{-0.25in} $\xi_1$}
%\psfrag{y}[][][0.9]{\qquad$\xi_2$}
%\psfrag{q}[][][0.9]{$q$}
%\psfrag{tau}[][][0.9]{$\tau$}
%\psfrag{z1}[][][0.9]{$z_1$}
%\psfrag{z0}[][][0.9]{$z_0$}
%\psfrag{flows [t]}[][][0.7]{$t$}
%\subfigure[Plant trajectory \label{fig:PlanarPlot-LimitedInfo3}]
%{\includegraphics[width=.47\textwidth]{plane3}}  
%\subfigure[Controller trajectory \label{fig:EstimatorPlot-LimitedInfo3}]
%{\includegraphics[width=.52\textwidth]{ControllerStates3}}  
%\end{center}  
%\caption{Plant and controller states of a closed-loop trajectory.  (a) Plant component $\xi(t,j)$ for \eqref{eqn:plantEx1} from $\xi(0,0)=(3,-3)$,
%$q(0,0) = 1$, $\tau(0,0) = z_0(0,0)= 0$, 
%$z_1(0,0) = 1$.  
%Dotted lines denote an estimate of ${\cal B}_0$, 
%$\star$ (red) the jump from $q=1$ to $0$,
%and $\times$ the sets $\A_1 = (0,\overline{\alpha})$ and $\A_0 = (0,0)$,
%with $\overline{\alpha} = \frac{1}{4}$.
%(b) Controller states of hybrid supervisor $\KK_s$.
%The dashed lines represent the jumps in the variables.
%Controller parameters: 
%$\eps_{0,a} = \frac{4}{27}$, $\eps_{1,a}= 0.00005$, and $\tau^* = 15$.
%}
%\label{fig:1bisbis}
%\end{figure}

%-----------------------
\begin{figure}[h!]  
\begin{center}  
\subfigure[Plant trajectory \label{fig:PlanarPlot-LimitedInfo3}]
{
\psfragfig*[width=.43\textwidth]{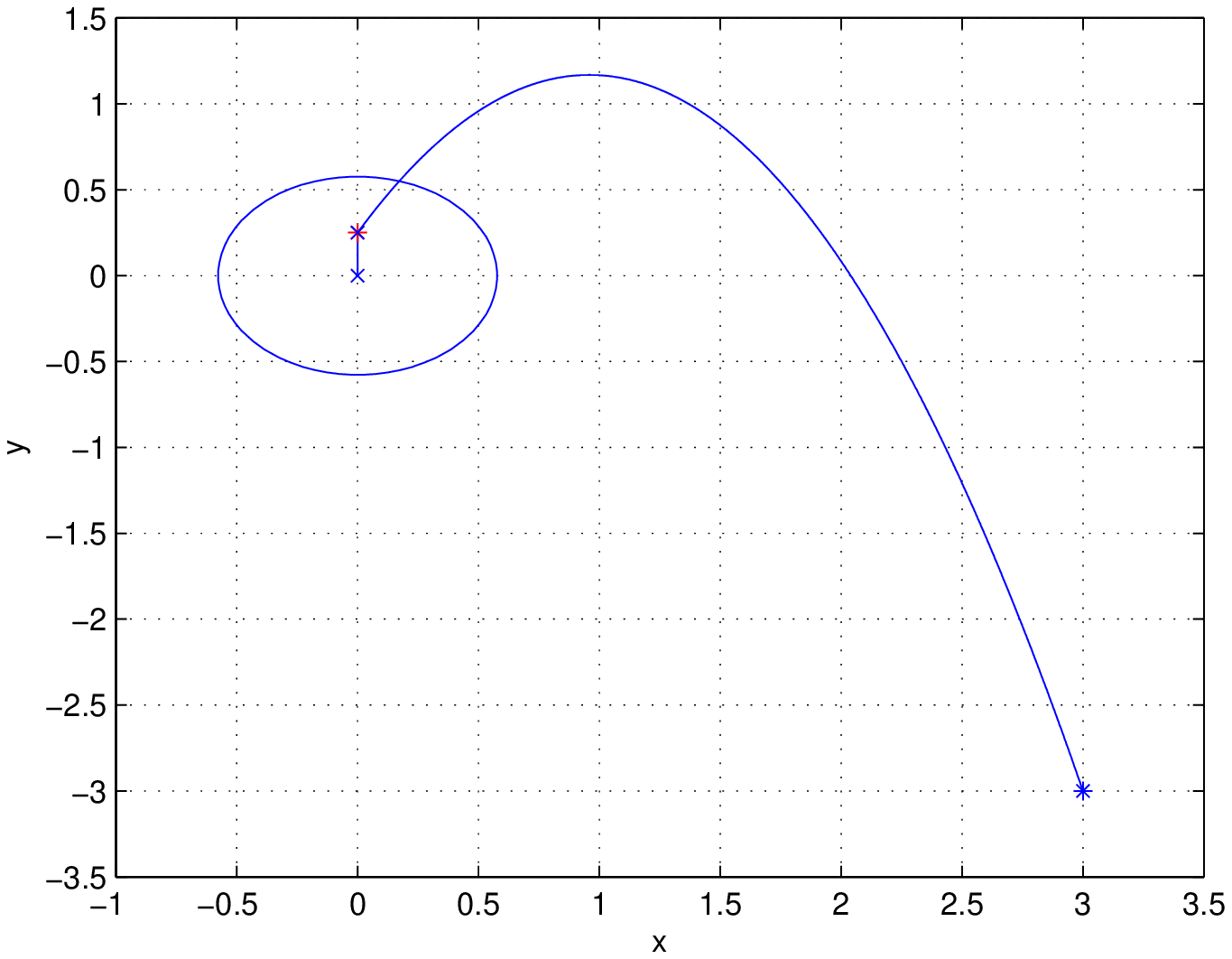}
{
\psfrag{x}[][][0.9]{\hspace{-0.25in} $\xi_1$}
\psfrag{y}[][][0.9]{\qquad$\xi_2$}
\psfrag{q}[][][0.9]{$q$}
\psfrag{tau}[][][0.9]{$\tau$}
\psfrag{z1}[][][0.9]{$z_1$}
\psfrag{z0}[][][0.9]{$z_0$}
\psfrag{flows [t]}[][][0.7]{$t$}
}
}  
\subfigure[Controller trajectory \label{fig:EstimatorPlot-LimitedInfo3}]
{
\psfragfig*[width=.48\textwidth]{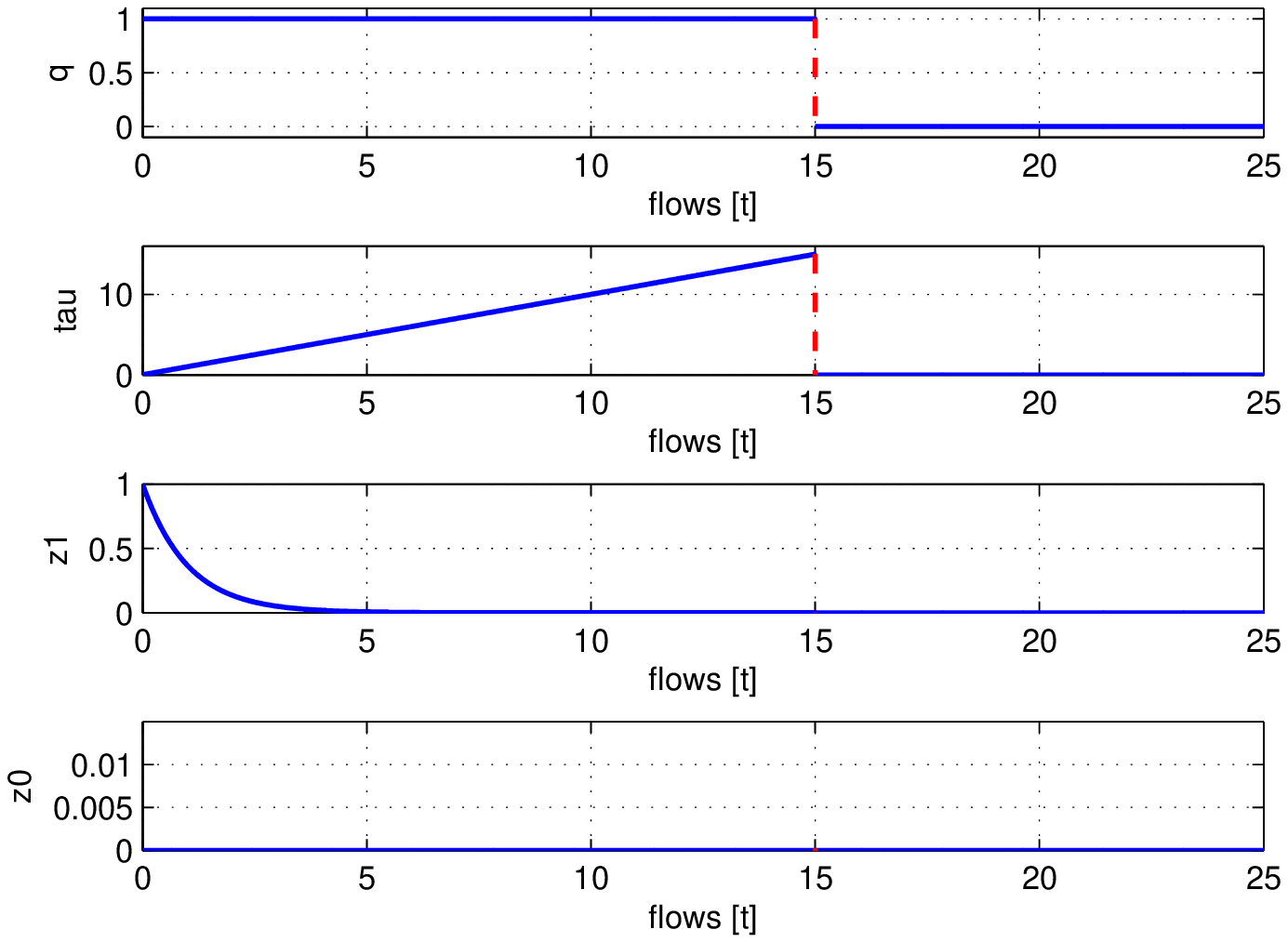}
{
\psfrag{x}[][][0.9]{\hspace{-0.25in} $\xi_1$}
\psfrag{y}[][][0.9]{\qquad$\xi_2$}
\psfrag{q}[][][0.9]{$q$}
\psfrag{tau}[][][0.9]{$\tau$}
\psfrag{z1}[][][0.9]{$z_1$}
\psfrag{z0}[][][0.9]{$z_0$}
\psfrag{flows [t]}[][][0.7]{$t$}
}
}  
\end{center}  
\caption{Plant and controller states of a closed-loop trajectory.  (a) Plant component $\xi(t,j)$ for \eqref{eqn:plantEx1} from $\xi(0,0)=(3,-3)$,
$q(0,0) = 1$, $\tau(0,0) = z_0(0,0)= 0$, 
$z_1(0,0) = 1$.  
Dotted lines denote an estimate of ${\cal B}_0$, 
$\star$ (red) the jump from $q=1$ to $0$,
and $\times$ the sets $\A_1 = (0,\overline{\alpha})$ and $\A_0 = (0,0)$,
with $\overline{\alpha} = \frac{1}{4}$.
(b) Controller states of hybrid supervisor $\KK_s$.
The dashed lines represent the jumps in the variables.
Controller parameters: 
$\eps_{0,a} = \frac{4}{27}$, $\eps_{1,a}= 0.00005$, and $\tau^* = 15$.
}
\label{fig:1bisbis}
\end{figure}

%-----------------------

\subsection{Stabilization under topological obstructions}
\label{ex:2}

Consider the stabilization of the point $\A_0 := \{\xi^*\}$,
for the point-mass system 
in Example~\ref{ex:Avoidance1}.
Following the discussions therein, 
 the measurements available are
\begin{equation}
\begin{array}{rcl}
y_1 &=&
h_1(\xi)
:=
\left(
\varphi_{1}(\xi),
\nabla \varphi_{1}(\xi),
\varphi_{2}(\xi),
\nabla \varphi_{2}(\xi)
\right)
\quad 
\forall \xi \in \reals^2,\\
y_2 &=&
h_2(\xi) :=  \xi
\qquad \qquad \quad 
\forall \xi \in \xi^*+\varepsilon\ball
\end{array}\end{equation}
for some $\varepsilon >0$,
where
$\varphi_i$, $i = 1,2$, are continuously differentiable functions given by
$$
\varphi_i(\xi) := \frac{1}{2} (\xi-\xi^{\circ})^{\top}(\xi-\xi^{\circ}) + B(d_i(\xi))
$$
with $B:\realsgeq\to\reals$ a continuously differentiable
function defined as $B(z):=\max\{0,(z-1)^2\ln\frac{1}{z}\}$ and
$d_i:\reals^2\to\realsgeq$ a continuously differentiable function
that measures the distance from any point in $O_i$ to the set ${\cal N}$.
These functions define ``potential'' functions 
relative to the intermediate target point $\xi^\circ$ that include the 
presence of the obstacle.
The sets $\mathcal{N}$ for $\hat{\alpha}= 0.07$ and $\overline{\xi} = (1,0)$, 
$\A_0$ for $\{\xi^*\} = \{(4,-\frac{1}{4})\}$, 
and $O_i$ given by
$
O_1 = \defset{\xi \in \reals^2}{|\xi_1| - 1.1 \geq \xi_2}$,
$O_2 = \defset{\xi \in \reals^2}{|\xi_1| + 1.1 \leq \xi_2}$
are depicted in Figure~\ref{fig:2}.
The point $\xi^{\circ}$ is the point at which $\varphi_i$ vanishes.
The local controller can measure the full state $\xi$ in the neighborhood $\A_0+ \eps\ball$ 
for $\eps = 1$.

%\begin{figure}[hb!]  
%\begin{center}  
%\psfrag{x1}[][][0.9]{\hspace{-0.25in} $\xi_1$}
%\psfrag{x2}[][][0.9]{\qquad$\xi_2$}
%\psfrag{ A}[][][0.9]{$\A_0$}
%\psfrag{ A1}[][][0.9]{$\A_1$}
%\psfrag{O1}[][][0.9]{$O_1$}
%\psfrag{O2}[][][0.9]{\!\!\!\!\!\!\!\!$O_2$}
%\subfigure[
%Plant trajectory with initial conditions $\xi(0,0)=0$,
%$q(0,0) = 1$, $\zeta_1(0,0)=1$, 
%steered below the obstacle
%using $\kappa_1(\xi,1)$ while in $\zeta_1 = 1$.
%\label{fig:PlanarPlot-Avoidance1}]
%{\includegraphics[width=.47\textwidth]{AvoidanceMissionHybridNew}}  \quad
%\subfigure[
%Plant trajectory  with initial conditions $\xi(0,0)=0$, $q(0,0) = 1$, $\zeta_1(0,0)=2$,
%steered above the obstacle
%using $\kappa_1(\xi,2)$ while in $\zeta_1 = 2$.
%\label{fig:PlanarPlot-Avoidance2}]
%{\includegraphics[width=.47\textwidth]{AvoidanceMissionHybridNew4}}  
%\end{center}  
%\caption{Trajectories $\xi(t,j)$ to point-mass system with hybrid supervisor $\KK_s$.
%Dotted circle denotes an estimate of ${\cal B}_0$
%and $\times$ the sets $\A_1 = \{(3,0)\}$ and $\A_0 = \{(4,-\frac{1}{4})\}$.
%The set $O_1$ is the region below the upper ``wedge,''
%while the set $O_2$ is  the region above the lower ``wedge,''
%which is depicted in dotted line.
%The cone emanating from the initial condition
%depicts that, initially, the target point is not in the
%line-of-sight of the point-mass system.
%The controller parameters used are $\mu = 1.1$ and $\lambda = 0.09$.
%}
%\label{fig:2}
%\end{figure}

%-----------------------
\begin{figure}[hb!]  
\begin{center}  
\subfigure[
Plant trajectory with initial conditions $\xi(0,0)=0$,
$q(0,0) = 1$, $\zeta_1(0,0)=1$, 
steered below the obstacle
using $\kappa_1(\xi,1)$ while in $\zeta_1 = 1$.
\label{fig:PlanarPlot-Avoidance1}]
{
\psfragfig*[width=.45\textwidth]{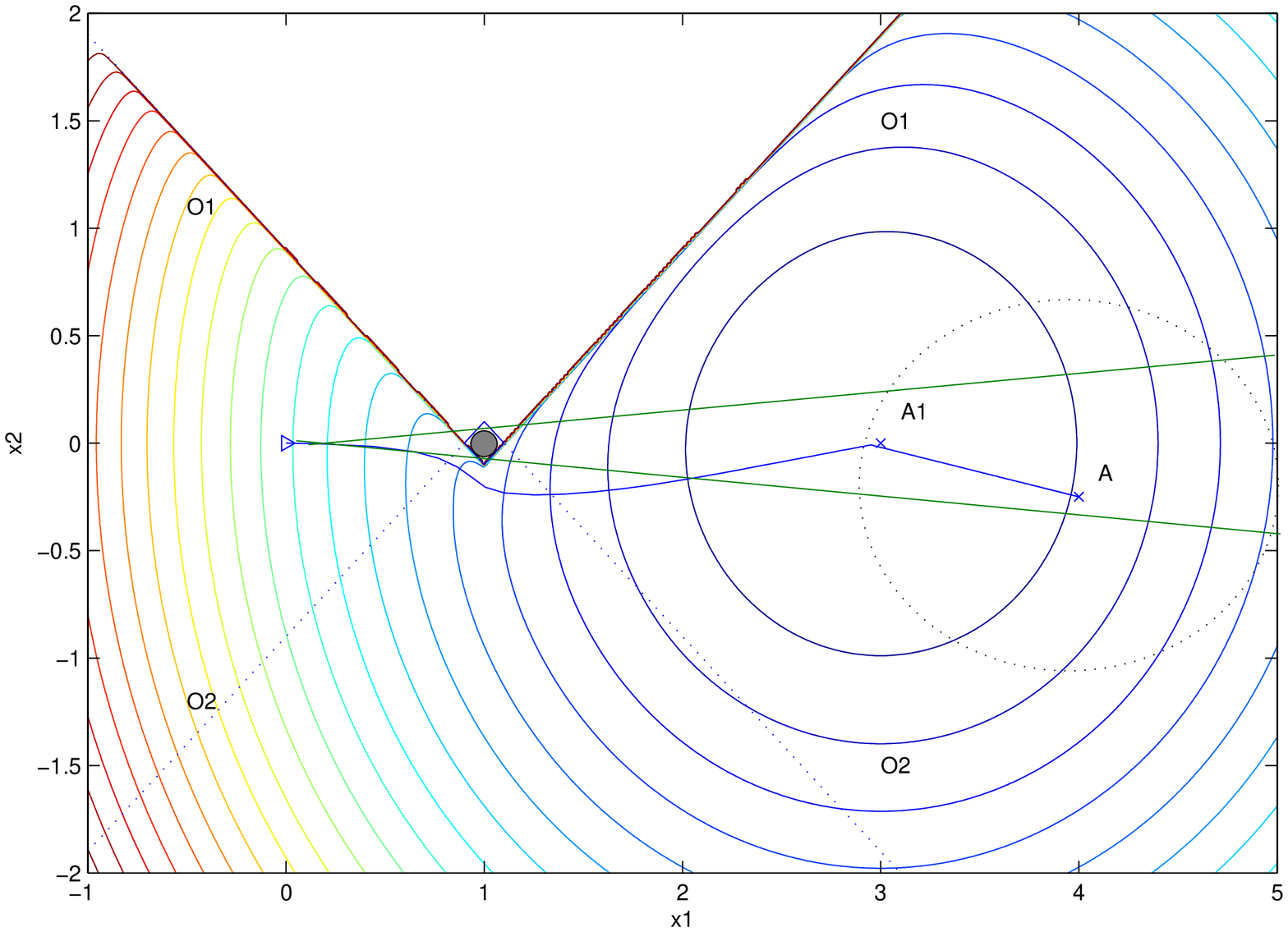}
{
\psfrag{x1}[][][0.9]{\hspace{-0.25in} $\xi_1$}
\psfrag{x2}[][][0.9]{\qquad$\xi_2$}
\psfrag{ A}[][][0.9]{$\A_0$}
\psfrag{ A1}[][][0.9]{$\A_1$}
\psfrag{O1}[][][0.9]{$O_1$}
\psfrag{O2}[][][0.9]{\!\!\!\!\!\!\!\!$O_2$}
}
}
  \quad
\subfigure[
Plant trajectory  with initial conditions $\xi(0,0)=0$, $q(0,0) = 1$, $\zeta_1(0,0)=2$,
steered above the obstacle
using $\kappa_1(\xi,2)$ while in $\zeta_1 = 2$.
\label{fig:PlanarPlot-Avoidance2}]
{
\psfragfig*[width=.45\textwidth]{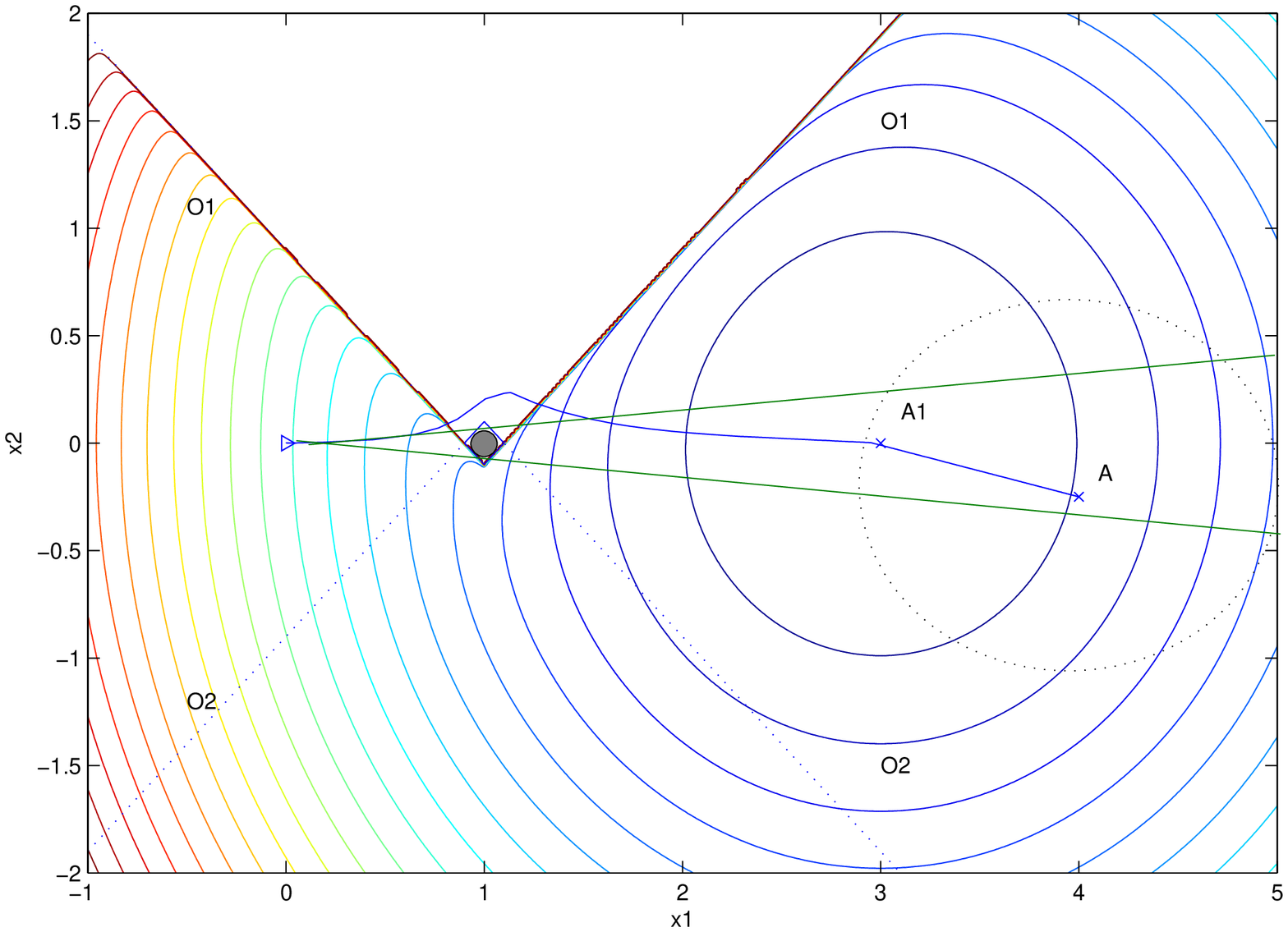}
{
\psfrag{x1}[][][0.9]{\hspace{-0.25in} $\xi_1$}
\psfrag{x2}[][][0.9]{\qquad$\xi_2$}
\psfrag{ A}[][][0.9]{$\A_0$}
\psfrag{ A1}[][][0.9]{$\A_1$}
\psfrag{O1}[][][0.9]{$O_1$}
\psfrag{O2}[][][0.9]{\!\!\!\!\!\!\!\!$O_2$}
}
}  
\end{center}  
\caption{Trajectories $\xi(t,j)$ to point-mass system with hybrid supervisor $\KK_s$.
Dotted circle denotes an estimate of ${\cal B}_0$
and $\times$ the sets $\A_1 = \{(3,0)\}$ and $\A_0 = \{(4,-\frac{1}{4})\}$.
The set $O_1$ is the region below the upper ``wedge,''
while the set $O_2$ is  the region above the lower ``wedge,''
which is depicted in dotted line.
The cone emanating from the initial condition
depicts that, initially, the target point is not in the
line-of-sight of the point-mass system.
The controller parameters used are $\mu = 1.1$ and $\lambda = 0.09$.
}
\label{fig:2}
\end{figure}

%-----------------------

We design a hybrid supervisor $\KK_s$ to coordinate
two output-feedback controllers.
The controller while in mode $q=1$ is hybrid  with a discrete state $\zeta_1 \in \{1,2\}$ 
evolving continuously according to
$\dot{\zeta}_1 = 0$.
The target stabilization set for this controller is taken to be $\A_1= \{ \xi^{\circ} \}$.
Let $\mu > 1, 
\lambda \in (0,\mu-1)$.
The following hybrid controller defines the feedback law $\KK_1$
$
\nonumber 
\kappa_{c,1}(\xi,\zeta_1) := -\nabla \varphi_{\zeta_1}(\xi)
$
when $(\xi,\zeta_1)\in  C_{c,1}$,
where
$\\[1em]
C_{c,1}
:=\{
(\xi,\zeta_1) \in \cup_{\zeta_1\in \{1,2\}}
(O_{\zeta_1}\times \{\zeta_1\})\ : \
\null\hfill$\\[0.5em]\null\hfill$
{\varphi_{\zeta_1}(\xi)\leq \mu\min_{\zeta_1\in \{1,2\}} \varphi_{\zeta_1}(\xi)}\}
$\\[0.7em]
and has discrete dynamics given by 
$$
\zeta_1^+ \in G_1(y_1,\zeta_1):=
  \defset{\zeta_1'\in \{1,2\}}{\varphi_{\zeta_1}(\xi) \geq (\mu-\lambda) \varphi_{\zeta_1'}(\xi)}
$$
when $(\xi,\zeta_1)\in D_{c,1}$,
where 
$$
\begin{array}{rcl}
D_{c,1} := \defset{
(y_1,\zeta_1)
}
{ \varphi_{\zeta_1}(\xi)\geq (\mu-\lambda)\min_{\zeta_1'\in \{1,2\}}
  \varphi_{\zeta_1'}(\xi)} .
\end{array}
$$
The design parameters of the controller $\KK_1$
are $\mu$
and $\lambda$.

Take $V(\xi,\zeta_1) = \varphi_{\zeta_1}(\xi)$, then with the $\KK_1$ dynamics we obtain,
with $\gamma':=(\mu- \lambda)^{-1}$, $\gamma' \in (0,1)$, $\rho(s) = s^2$,
\begin{equation}\non
\begin{array}{rcl}
V(\xi,\zeta_1') \leq \gamma' V(\xi,\zeta_1)  \quad \forall \zeta_1' \in G_1(\xi,\zeta_1), \ \forall (\xi,\zeta_1) \in D_{c,1} \ ,
\end{array}
\end{equation}
and, $\forall (\xi,\zeta_1) \in C_{c,1}$,
\begin{equation}\non
\begin{array}{rcl}
 \langle \nabla V(\xi,\zeta_1') , f_p(\xi,\kappa_1(\xi,\zeta_1)) \rangle
\leq - 2\,V(\xi,\zeta_1)\ .
\end{array}
\end{equation}
Global asymptotic stability of $\A_1$ (on $C_{c,1} \cup D_{c,1}$)
follows, 
from where a norm observer for $|\xi|_{\A_1}$ exists;
e.g., 
we can use 
$\eps_1 = 1-\gamma'$ and any class-$\cal K$ function $\gamma_{1}$ for the norm observer in
\eqref{eqn:NormEstimator}.
The local controller to use in mode $q=0$
is a static, continuous-time feedback of the form
$\kappa_{c,0}(\xi) := -\xi + \xi^*$.
Local asymptotic stability of $\A_0$ follows with basin of attraction
$\A_0+\eps\ball$ and $\dot z_0 = -z_0$ is a norm observer for $|\xi|_{\A_0}$.

Figure~\ref{fig:2} depicts trajectories to the plant with
the proposed hybrid supervisor 
for two different initial conditions of the state $\zeta_1$
of the  controller $\KK_1$.
The trajectories converge first to a neighborhood of $\A_1$,
and when $z_1$ becomes small enough, 
a jump to $\KK_0$ is triggered 
and the trajectories converge to $\A_0$.

\section{Conclusion}
\label{sec:conclusion}

A solution to a general uniting problem was formulated and exercised
in examples.
The controllers considered can be hybrid, nonlinear, output-feedback, and have different objectives.
The solution consists of constructing a well-posed hybrid supervisor
that appropriately combines two hybrid controllers to accomplish the task.
In addition to stability and attractivity properties, to guarantee the existence
of norm estimators, the individual
controllers are assumed to induce an output-to-state stability property.
Robustness of the full closed-loop system is asserted via
results for perturbed hybrid systems.
Examples illustrating the design methodology of the hybrid supervisor
were presented.  The proposed algorithm can also be used for waypoint navigation 
and loitering control of unmanned aerial vehicles \cite{Smith.Sanfelice.13.GNC}.
The proposed solution does not assume a detectability property for the 
plant
 and thus, in contrast to \cite{PrieurTeel:ieee:11},  a 
global norm observer may not exist. 
When this stronger property is assumed, 
the proposed hybrid supervisor achieves robust, global asymptotic stability.
Moreover, the attractivity property in Assumption \ref{assumption:1} can be 
relaxed to a semi-global, practical attractivity property.
\balance

\bibliographystyle{plain}
\bibliography{long,Biblio,RGS,added-cp}

\end{document}